\documentclass[10pt]{article}
\usepackage[top=0.9in,bottom=1.1in,left=1.05in,right=1.05in]{geometry}
\usepackage{amsmath,amssymb}
\usepackage{amsthm}
\usepackage{latexsym}
\usepackage{mathrsfs}
\usepackage{cancel}
\usepackage{comment}
\usepackage{float}
\usepackage{graphicx}
\usepackage{epstopdf}
\usepackage{color}
\usepackage{enumerate}
\usepackage{natbib}
\DeclareGraphicsExtensions{.eps}

\numberwithin{equation}{section}

\newcommand{\MCC}{\mathcal{C}}
\newcommand{\MCA}{\mathcal{A}}
\newcommand{\MCM}{\mathcal{M}}
\newcommand{\MCO}{\mathcal{O}}

\newcommand{\MCH}{\mathcal{H}}
\newcommand{\MCL}{\mathcal{L}}
\newcommand{\MCK}{\mathcal{K}}
\newcommand{\EE}{\mathbb{E}}

\newcommand{\RR}{\mathbb{R}}
\newcommand{\NN}{\mathbb{N}}

\newcommand{\Ltx}{\mathcal{L}_{t,x}}

\newcommand{\Vz}{V^{\delta}}
\newcommand{\rz}{R}
\newcommand{\vz}{v^{(0)}}
\newcommand{\vo}{v^{(1)}}
\newcommand{\vt}{v^{(2)}}

\newcommand{\pz}{{\pi^{(0)}}}

\newcommand{\Vzl}{V^{\pz,\delta}}

\newcommand{\pzt}{\widetilde\pi^0}
\newcommand{\pot}{\widetilde\pi^1}

\newcommand{\Vzt}{\widetilde{V}^\delta}
\newcommand{\vzt}{\widetilde{v}^{(0)}}
\newcommand{\vot}{\widetilde{v}^{(1)}}

\newcommand{\vat}{\widetilde{v}^\alpha}
\newcommand{\vtat}{\widetilde{v}^{2\alpha}}
\newcommand{\vthat}{\widetilde{v}^{3\alpha}}

\newcommand{\abs}[1]{\left|#1\right|}
\newcommand{\average}[1]{\left\langle#1\right\rangle}

\newcommand{\mc}[1]{\mathcal{#1}}

\newcommand{\ud}{\,\mathrm{d}}
\newcommand{\nud}[1]{\,\nu(\mathrm{d}{#1})}

\newtheorem{theo}{Theorem}[section]
\newtheorem{lem}[theo]{Lemma}

\newtheorem{rem}[theo]{Remark}
\newtheorem{prop}[theo]{Proposition}
\newtheorem{assump}[theo]{Assumption}
\newtheorem{cor}[theo]{Corollary}

\newtheorem{theo*}{Theorem}[section]
\newtheoremstyle{dotlessS}{}{}{\color{blue}}{}{\color{blue}\bfseries}{}{ }{}
\theoremstyle{dotlessS}

\newcommand{\done}[1]{{\leavevmode\color{black}{#1}}}
\begin{document}

\title{\vspace{-50pt} Asymptotic Optimal Strategy for Portfolio Optimization in a Slowly Varying Stochastic Environment}
\author{Jean-Pierre Fouque\thanks{Department of Statistics \& Applied Probability,
 University of California,
        Santa Barbara, CA 93106-3110, {\em fouque@pstat.ucsb.edu}. Work  supported by NSF grant DMS-1409434.}
        \and Ruimeng Hu\thanks{Department of Statistics \& Applied Probability,
 University of California,
        Santa Barbara, CA 93106-3110, {\em hu@pstat.ucsb.edu}.}
        }
\date{\today}
\maketitle

\begin{abstract}

In this paper, we study the portfolio optimization problem with general utility functions and when the return and volatility of the underlying asset are slowly varying. An asymptotic optimal strategy is provided within a specific class of admissible controls under this problem setup. Specifically, we establish a rigorous first order approximation of the value function associated to a fixed zeroth order suboptimal trading strategy, which is derived heuristically in [J.-P. Fouque, R. Sircar and T. Zariphopoulou, {\it Mathematical Finance}, 2016]. Then, we show that this zeroth order suboptimal strategy is asymptotically optimal in a specific family of admissible trading strategies. Finally, we show that our assumptions are satisfied by a particular fully solvable model.

\end{abstract}


\section{Introduction}\label{sec_intro}

The portfolio optimization problem was first introduced and studied in the continuous-time framework in \cite{Me:69, Me:71}, which provided explicit solutions on how to trade stocks and/or how to consume so as to maximize one's utility, with risky assets following the Black-Scholes-Merton model (that is, geometric Brownian motions with constant returns and constant volatilities), and when the utility function is of specific types (for instance, {Constant Relative Risk Aversion} (CRRA)). Following these pioneer works, additional constraints were added in this model to mimic real-life investments. This includes transaction cost originally considered by \cite{MaCo:76} and a user's guide by \cite{GuMu:13}, and investments under drawdown constraint for instance by  \cite{GrZh:93},  \cite{CvKa:95} and  \cite{ElTo:08}, just to name a few. Using duality, \cite{CoHu:89} and  \cite{KaLeSh:87} studied the incomplete market case, and general analysis for semi-martingale models is provided in \cite{KrSc:03}. The case where the drift  and  volatility terms are non-linear functions of the asset price  is studied in \cite{Za:99}, and \cite{ChVi:05} gave a closed-form solution under a particular one-factor stochastic volatility model. 

Recently, multiscale factor models for risky assets were considered in the portfolio optimization problem in \cite{FoSiZa:13}, where return and volatility are driven by fast and slow factors.
Specifically, the authors heuristically derived the asymptotic approximation to the value function and the optimal strategy
for general utility functions. In this paper, we shall focus on the risky asset modeled by only \emph{slowly varying stochastic factor}, and the reason is twofold: Firstly, slow factor is particularly important in long-term investment, because the effect of fast factor is approximately averaged out in the long time as studied in \cite[Section 2]{FoSiZa:13}. Secondly, analysis under the model with \emph{fast mean-reverting stochastic factor} requires singular asymptotic techniques, and more technical details in combining the fast and slow factors, and thus, this will be presented in another paper in preparation (\cite{FoHu:XX}).

We describe the model as below, with dynamics of the underlying asset and slowly varying factor denoted as $S_t$ and $Z_t$ respectively,
\begin{align}\label{eq_St}
&\ud S_t = \mu(Z_t)S_t\ud t + \sigma(Z_t)S_t\ud W_t, \\
&\ud Z_t = \delta c(Z_t) \ud t + \sqrt\delta g(Z_t)\ud W_t^Z \label{eq_Zt},
\end{align}
where the standard Brownian motions $\left(W_t, W_t^Z\right)$ are correlated:
$\ud \average{W, W^Z}_t = \rho  \ud t,$ with $ \abs{\rho } < 1$.

Assumptions on the coefficients $\mu(z), \sigma(z), c(z), g(z)$ of the model will be specified in Section~\ref{sec_assumpvaluefunc}. In \eqref{eq_Zt}, $\delta$ is a small positive parameter that characterizes the slow variation of the process $Z$. Note that $Z_t \stackrel{\mathcal{D}}{=}Z^{(1)}_{\delta t}$, as continuous processes  with $\forall t\in [0,T]$,  where the diffusion process $Z^{(1)}$ has the following infinitesimal generator, denoted by $\MCM$
\begin{equation}\label{eq_Mgenerator}
\MCM = \frac{1}{2}g^2(z)\partial_{zz}^2 + c(z)\partial_z.
\end{equation}
We refer to \cite{FoPaSiSo:11} for more details on this model, where asymptotic results in the limit $\delta\to 0$ are derived for linear problems of option pricing.

Denote by $X_t^\pi$ the wealth process associated to the Markovian strategy $\pi$, and in this strategy, the amount of money $\pi(t,x,z)$ is invested in stock at time $t$, when the stock price is $x$, and the level of the slow factor $Z_t$ is $z$, with the remaining money held in money market earning a constant risk-free interest of $r$. Assuming that the portfolio is self-financing and, without loss of generality, that the risk-free interest rate is zero, $r=0$, then $X_t^\pi$ follows
\begin{equation}\label{eq_Xtgeneral}
\ud X_t^\pi = \pi(t,X_t^\pi,Z_t)\mu(Z_t)\ud t + \pi(t,X_t^\pi,Z_t)\sigma(Z_t)\ud W_t.
\end{equation}

An investor aims at finding an optimal strategy $\pi$  which maximizes her  terminal expected utility $\EE\left[U(X_T^\pi)\right]$, where $U(x)$ is in a general class of utility functions. Denote by $\Vz(t,x,z)$ the value function
\begin{equation}\label{def_Vz}
\Vz(t,x,z) = \sup_{\pi\in \MCA^\delta(t,x,z)}\EE\left[U(X_T^\pi)|X_t^\pi = x, Z_t = z\right],
\end{equation}
where the supremum is taken over all admissible Markovian  strategies $\MCA^\delta(t,x,z)$,
\begin{equation}\label{def_A}
\MCA^\delta(t,x,z)= \left\{ \pi: X_s^\pi \text{ in \eqref{eq_Xtgeneral} stays nonnegative $\forall s \geq t$, given $X_t^\pi = x$, and $Z_t = z$}\right\}.
\end{equation}
%

 In \cite{FoSiZa:13}, a regular perturbation approach is used to derive a heuristic approximation for $\Vz$ up to the first order, namely, the value function $\Vz$ is formally expanded  as follows:
\begin{equation}\label{expanofVdelta}
\Vz = \vz + \sqrt\delta\vo + \delta\vt + \cdots,
\end{equation}
with $\vz$ and $\vo$ identified by asymptotic equations given in Section \ref{sec_heuristic}. \done{Note that this derivation is rigorous in the case of power utility with one factor stochastic volatility, which is done in \cite[Section 6.3]{FoSiZa:13}.} It is also \done{conjectured} in \cite[Section 3.2.1]{FoSiZa:13} that the zeroth order suboptimal strategy
\begin{equation}\label{pi0}
\pz(t,x,z) = -\frac{\lambda(z)}{\sigma(z)}\frac{\vz_{x}(t,x,z)}{\vz_{xx}(t,x,z)}, \quad \lambda(z) = \frac{\mu(z)}{\sigma(z)},
\end{equation}
not only gives the optimal value at the principal term $\vz$, but also up to first order $\sqrt\delta$ correction $\vz + \sqrt\delta \vo$.  

\done{The optimal control to problem \eqref{def_Vz}, denoted by $\pi^\ast$, whose existence is ensured by \cite{KrSc:03}, depends on $\delta$. It is not known whether $\pi^\ast$ will converge as $\delta$ goes to zero. But if $\pi^\ast$ had a limit, say $\pzt$, it is then natural to consider a family of controls of the form $\pzt + \delta^\alpha \pot$ as the perturbation of the limit $\pzt$, allowing for correction of any order in $\delta$. }
	
\done{ The goal of this paper is to show that $\pz$ given by \eqref{pi0} in fact performs asymptotically better than the family $\pzt + \delta^\alpha \pot$ up to the order $\sqrt{\delta}$. To this end, for a fixed choice of $(\pzt$, $\pot$, $\alpha>0)$, we introduce the family of admissible trading strategies $\MCA_0(t,x,z)\left[\pzt, \pot,\alpha\right]$ defined by
	\begin{equation}\label{def_A0}
	\MCA_0(t,x,z)\left[\pzt, \pot,\alpha\right] = \left\{ \pzt + \delta^\alpha \pot\right\}_{0 < \delta \leq 1}.
	\end{equation}
Further conditions will be given in Assumption \ref{assump_piregularity}. Denote by $\Vzl$ the value function associated to the strategy $\pz$, that is 
\begin{equation*}
\Vzl(t,x,z) = \EE\left[U(X_T^\pz)|X_t^\pz = x, Z_t = z\right], 
\end{equation*}
where $X_t^\pz$ is given by \eqref{eq_Xtgeneral} with $\pi = \pz$. Let us also denote $\Vzt$ the value function when using the strategy $\pzt + \delta^\alpha \pot$, that is,
\begin{equation*}
\Vzt(t,x,z) = \EE\left[U(X_T^{\pzt + \delta^\alpha\pot})|X_t^{\pzt + \delta^\alpha\pot} = x, Z_t = z\right].
\end{equation*}

Our main result is:}

\begin{theo*}\label{Thm_optimality}
	Under Assumption \ref{assump_U}, \ref{assump_valuefunc}, \ref{assump_piregularity} and \ref{assump_optimality}, for fixed $(t,x,z)$ and any family of trading strategies $\MCA_0(t,x,z)\left[ \pzt, \pot,\alpha\right]$, the following limit exists and satisfies
	\begin{equation}\label{eq_Vztineq}
	\ell := \lim_{\delta\to 0}\frac{\Vzt(t, x, z)-\Vzl(t, x, z)}{\sqrt{\delta}}\leq 0.
	\end{equation}
	That is, the strategy $\pz$ which generates $\Vzl$, performs  asymptotically better up to order $\sqrt\delta$
	than the family $\left\{\pzt+\delta^\alpha\pot\right\}$ which generates $\Vzt$.
	
	Moreover, the inequality can be written according to the following four cases:
	\begin{enumerate}[(i)]
		\item$\pzt \equiv \pz$ and $\ell = 0$: $\Vzt = \Vzl + o(\sqrt{\delta})$;
		\item $\pzt \equiv \pz$ and $-\infty < \ell <0$: $\Vzt = \Vzl + \MCO(\sqrt{\delta})$ with $\MCO(\sqrt{\delta}) <0$;
		\item $\pzt \equiv \pz$ and $\ell = -\infty$: $\Vzt = \Vzl + \MCO(\delta^{2\alpha})$ with $\MCO(\delta^{2\alpha})<0$ and $2\alpha < 1/2$;
		\item $\pzt \not \equiv \pz$: $
		\lim_{\delta \to 0} \Vzt(t,x,z)< \lim_{\delta \to 0}  \Vzl(t,x,z).$
		\end{enumerate}
\end{theo*}

\done{\begin{rem}\label{rem:pi1}
	In particular, choosing $\pzt \equiv \pz$, $\alpha=1/2$ and $\pot \equiv \pi^1$ where $\pi^1$ is the first order correction in the expansion of the strategy derived in   \cite[Section 3.2.2]{FoSiZa:13}, we obtain that this first order correction $\pi^1$ does not affect the value function up to order $\sqrt{\delta}$.
\end{rem}}

\medskip

\noindent\done{{\bf Organization of the paper.} In Section~\ref{sec:assumptions}, we briefly review the classical Merton problem and heuristic results in \cite{FoSiZa:13}. We also list the assumptions needed for our theoretical proofs given in the next Section, where we apply the regular perturbation technique to the value function $\Vzl$ associated to the strategy $\pz$, and where we prove that the first order approximation of $\Vzl$ is $\vz + \sqrt{\delta}\vo$ for general utility function $U(x)$. In the case of power utility, we show in Corollary~\ref{cor_optimality} that $\pz$ is in fact asymptotically optimal in the full class of admissible strategies $\MCA^\delta$. (Note that this result remains an open problem for general utilities.) Then, the proof of Theorem~\ref{Thm_optimality} is given in Section~\ref{sec_optimality}. A fully solvable example is presented in Section~\ref{sec_example}. Using the model given in  \cite[Section 6.4]{FoSiZa:13},  we give the closed-form solution under power utility and we verify that all the assumptions listed in Section~\ref{sec:assumptions} and Section~\ref{sec_optimality} are satisfied. We make conclusive remarks in Section~\ref{sec_conclusion}.}

%
%
%
%
%

\section{Preliminaries and Assumptions}\label{sec:assumptions}

In this section, we review the classical Merton problem,  summarize the heuristic results in \cite{FoSiZa:13}, and  list the assumptions on the utility function and the state processes needed for later proofs.

\subsection{Merton Problem with Constant Coefficients}\label{sec_merton}
We first discuss the case of $\mu$ and $\sigma$ being constant in \eqref{eq_St}, which plays a crucial role in interpreting the leading order value function $\vz$ in \eqref{expanofVdelta} and analysis of the regular perturbation. This problem has been widely studied and completely solved, and we start with some background results.

Let $X_t$ be the solution to
	\begin{equation}
	\ud {X}_t = \pi^\star(t,{X}_t)\mu\ud t + \pi^\star(t,{X}_t)\sigma\ud W_t,
	\end{equation}	
	where $\pi^\star(t,x)$ is the optimal trading strategy, then, $X_t$ stays nonnegative up to time $T$, and
	\begin{equation}
	\int_0^T \abs{\sigma\pi^\star(t,X_t)}^2\ud t < \infty, \text{ almost surely.}
	\end{equation}
	We refer to \cite[Chapter 3]{KaSh:98} for details.

Following the notations in \cite{FoSiZa:13}, we denote by $M(t,x;\lambda)$ the Merton value function. The following results give the regularity of $M(t,x;\lambda)$ and identify it as the classical solution of an HJB equation. 
\begin{prop}\label{prop_Merton}
Assume that the utility function $U(x)$ is $C^2(0,\infty)$, strictly increasing, strictly concave, such that $U(0+)$ is finite, and satisfies the Inada and Asymptotic Elasticity conditions:
	\begin{equation*}
	U'(0+) = \infty, \quad U'(\infty) = 0, \quad \text{AE}[U] := \lim_{x\rightarrow \infty} x\frac{U'(x)}{U(x)} <1,
	\end{equation*}
	then, the Merton value function $M(t,x;\lambda)$ is strictly increasing, strictly concave in the wealth variable $x$, and decreasing in the time variable $t$. It is $C^{1,2}([0,T]\times \RR^+)$ and is the unique solution to the HJB equation
	\begin{equation}\label{eq_value}
	M_t+\sup_{\pi}\left\{\frac{1}{2}\sigma^2\pi^2M_{xx}+\mu\pi M_x\right\}=
	M_t -\frac{1}{2}\lambda^2\frac{M_x^2}{M_{xx}} = 0, \quad M(T,x;\lambda) = U(x),
	\end{equation}
	where $\lambda = \mu/\sigma$ is the constant Sharpe ratio. It is  $C^1$ with respect to $\lambda$, and 
	\begin{equation}\label{eq_pistar}
	\pi^\star(t,x;\lambda)=-\frac{\lambda}{\sigma}\frac{M_x(t,x;\lambda)}{M_{xx}(t,x;\lambda)}.
	\end{equation}
\end{prop}

\done{This proof without uniqueness follows from \cite[Section 3.8]{KaSh:98}. By the result in \cite{KaZa:14}, $M_x$ can be transformed into the unique solution of heat equation (This result will be stated in Proposition \ref{prop_H} and used in the rest of the paper). By the assumption that $U(0+)$ is finite, the uniqueness of $M(t,x;\lambda)$ follows.}

%
The following  relation between partial derivatives of $M(t,x;\lambda)$ is derived in \cite[Lemma 3.2]{FoSiZa:13}.
\begin{lem}\label{lem_vegagamma}
	The Merton value function $M(t,x;\lambda)$ satisfies the ``Vega-Gamma'' relation
	\begin{equation}
	M_\lambda = -(T-t)\lambda R^2 M_{xx},
	\end{equation}
 where
 \begin{equation}\label{def_risktolerance}
R(t,x;\lambda) = -\frac{M_x(t,x;\lambda)}{ M_{xx}(t,x;\lambda)},
\end{equation}
is  the {risk-tolerance function}.
\end{lem}
Note that $R(t,x;\lambda)$ is continuous, strictly positive due to the regularity, concavity and monotonicity of $M (t,x;\lambda)$. As introduced in \cite{FoSiZa:13}, we recall the notation
\begin{align}\label{def_dk}
D_k &= R(t,x;\lambda)^k \partial_x^k, \qquad k = 1,2, \cdots,\\
\Ltx(\lambda) &= \partial_t + \frac{1}{2}\lambda^2D_2 + \lambda^2D_1.\label{def_ltx}
\end{align}
Note that the coefficients of $\Ltx(\lambda)$ depend on $R(t,x;\lambda)$, and then on $M(t,x;\lambda)$, and the Merton PDE \eqref{eq_value} can be re-written as
\begin{align}\label{eq_mertonlinear}
&\Ltx(\lambda)M(t,x;\lambda) = 0.
\end{align}
The following  result regarding the linear operator $\Ltx(\lambda)$ will be used repeatedly in Sections~\ref{sec_accuracy} and \ref{sec_optimality}.
\begin{prop}\label{prop_ltxunique}
	Let $\Ltx(\lambda)$ be the operator defined in \eqref{def_ltx}, and assume that the utility function $U(x)$ satisfies the conditions in Proposition \ref{prop_Merton} and $U(0+) = 0$ (or finite), then 
	\begin{equation}\label{def_ltxpde}
	 \Ltx(\lambda)u(t,x;\lambda) = 0,\quad
	u(T,x;\lambda) = U(x),
	\end{equation}
	has a unique nonnegative solution.
\end{prop}
\begin{proof}
	First, observe that $M(t,x;\lambda)$ is a solution of \eqref{def_ltxpde}. To show uniqueness,   we use the following transformation
	\begin{align}\label{eq_trans}
	\left\{ {\begin{array}{*{20}{l}}
		{\xi = -\log M_x(t,x;\lambda)} + \frac{1}{2}\lambda^2(T-t), \\
		{t' = t},
		\end{array}} \right.
	\end{align}
	which is one-to-one since the Jacobian $\abs{\frac{-M_{xx}}{M_x}}$ stays positive. Define $w(t',\xi;\lambda) = u(t,x;\lambda)$, then $w$ solves:
	\begin{align*}
	\MCH w = w_{t'} + \frac{1}{2}\lambda^2 w_{\xi\xi} = 0, \quad w(T,\xi;\lambda) = U(I(e^{-\xi})).
	\end{align*}
	Uniqueness of the nonnegative solution then follows from classical results for the heat equation \cite[Chapter 7.1(d)]{JO:82}.
\end{proof}

\subsection{Heuristic Expansion under Slowly Varying Stochastic Factor}\label{sec_heuristic}
In this subsection, we summarize the expansion derived in  \cite{FoSiZa:13} that will be used in following sections. \done{The Hamilton-Jacobi-Bellman (HJB) equation  for $\Vz$ is given by
\begin{equation}\label{eq_Vz}
\Vz_t + \delta\MCM\Vz + \max_{\pi \in \MCA^\delta}\left(\frac{1}{2}\sigma(z)^2\pi^2\Vz_{xx} + \pi\left(\mu(z)\Vz_x + \sqrt\delta \rho  g(z)\sigma(z) \Vz_{xz}\right)\right) = 0.
\end{equation}
In general, $\Vz$ is only described as the viscosity solution of this HJB equation. In the heuristic derivation,  it is assumed that $\Vz$ is the unique classical solution of \eqref{eq_Vz}.
Now, maximizing in $\pi$ and plugging in the optimizer gives the following nonlinear equation with terminal condition,
\begin{equation}\label{eq_Vznonlinear}
\Vz_t + \delta \MCM\Vz - \frac{\left(\lambda(z)\Vz_x + \sqrt\delta\rho  g(z)\Vz_{xz}\right)^2}{2\Vz_{xx}} = 0,\quad
 \Vz(T,x,z) = U(x),
\end{equation}
where the optimizer (optimal control) is given in the feedback form by
\begin{equation*}
\pi^\ast = -\frac{\lambda(z)\Vz_x}{\sigma(z)\Vz_{xx}} - \frac{\sqrt{\delta}\rho  g(z)\Vz_{xz}}{\sigma(z)\Vz_{xx}},
\end{equation*}
and the Sharpe ratio is $\lambda(z) = \mu(z)/\sigma(z)$.
\begin{rem}
 In this paper, we do not assume regularity of $\Vz$ defined by \eqref{def_Vz}, as we will work with the quantity $\Vzl$ defined by \eqref{def_Vzl}, which will be the classical solution of the linear PDE \eqref{eq_Vzl}.
\end{rem}
The HJB equation \eqref{eq_Vznonlinear} is fully nonlinear and not explicitly solvable in general. The heuristic derivation of the expansion of $\Vz$ is obtained by applying the regular perturbation method, which consists substituting the expansion $\Vz = \vz + \sqrt\delta\vo + \delta\vt + \cdots$ into \eqref{eq_Vznonlinear} and collecting terms of corresponding orders to obtain:}
\begin{enumerate}[(i)]
	\item The \emph{leading order term} $\vz$ \done{is defined as the solution to} the Merton  PDE associated with the current (or frozen) Sharpe ratio:
		\begin{equation}
		\vz_t - \frac{1}{2}\lambda(z)^2\frac{\left(\vz_x\right)^2}{\vz_{xx}} = 0, \quad \vz(T,x,z) = U(x).\label{eq_vz}
		\end{equation}
   By uniqueness in Proposition \ref{prop_Merton}, we have
   	\begin{equation}\label{eq_vzandmerton}
   	\vz(t,x,z) = M\bigl(t,x;\lambda(z)\bigr),
   	\end{equation}

	\item The \emph{first order correction} $\vo$ \done{is defined as the solution to} the linear PDE:
	\begin{equation}
	\vo_t + \frac{1}{2}\lambda(z)^2\left(\frac{\vz_x}{\vz_{xx}}\right)^2\vo_{xx} - \lambda(z)^2 \frac{\vz_x}{\vz_{xx}}\vo_x = \rho \lambda(z)g(z)\frac{\vz_x\vz_{xz}}{\vz_{xx}},\quad \vo(T,x,z)= 0.\nonumber
	\end{equation}
	Using the notations \eqref{def_dk}-\eqref{def_ltx}, $\vo$ satisfies the following linear equation which admits a unique solution
	\begin{equation}\label{eq_vo}
	\Ltx(\lambda(z))\vo = -\rho \lambda(z)g(z)D_1\vz_z, \quad \vo(T,x,z) = 0.
	\end{equation}
	\item By the ``Vega-Gamma'' relation stated in Lemma \ref{lem_vegagamma}, the $z$-derivative of the leading order term $\vz$ satisfies:
	\begin{equation}\label{eq_vzgamma}
	\vz_z = -(T-t)\lambda(z)\lambda'(z)D_2 \vz,
	\end{equation}
	and $\vo$ is explicitly given in term of $\vz$ by
	\begin{equation}\label{eq_vzandvo}
	\vo = -\frac{1}{2}(T-t)\rho \lambda(z)g(z)\frac{\vz_{x}\vz_{xz}}{\vz_{xx}}.
	\end{equation}
\end{enumerate}

\subsection{Assumptions on the Utility $U(x)$}\label{sec_assumpU}

\begin{assump}\label{assump_U}
Throughout the paper, we make the following assumptions on the utility $U(x)$:
\begin{enumerate}[(i)]
  \item\label{assump_Uregularity}  U(x) is $C^7(0,\infty)$, strictly increasing, strictly concave and satisfying the following conditions (Inada and Asymptotic Elasticity):
\begin{equation}\label{eq_usualconditions}
U'(0+) = \infty, \quad U'(\infty) = 0, \quad \text{AE}[U] := \lim_{x\rightarrow \infty} x\frac{U'(x)}{U(x)} <1.
\end{equation}
\item\label{assump_Ubddbelow}U(0+) is finite. Without loss of generality, we assume U(0+) = 0.
  \item\label{assump_Urisktolerance} Denote by $R(x)$ the risk tolerance, 
\begin{equation}\label{eq_risktolerance}
  R(x) := -\frac{U'(x)}{U''(x)}.
\end{equation}
Assume that $R(0) = 0$, R(x) is strictly increasing and $R'(x) < \infty$ on $[0,\infty)$, and there exists $K\in\RR^+$, such that for $x \geq 0$, and $ 2\leq i \leq 5$,
      \begin{equation}\label{assump_Uiii}
      \abs{\partial_x^{(i)}R^i(x)} \leq K.
      \end{equation}
  \item\label{assump_Ugrowth} Define the inverse function of the marginal utility $U'(x)$ as $I: \RR^+ \to \RR^+$, $I(y) = U'^{(-1)}(y)$, and assume that, for some positive $\alpha$, $I(y)$ satisfies the polynomial growth condition:
\begin{equation}\label{cond_I}
I(y) \leq \alpha + \kappa y^{-\alpha}.
\end{equation}
\end{enumerate}
\end{assump}

Note that the risk tolerance $R(x)$ given by \eqref{eq_risktolerance} is in fact the risk tolerance function $R(t,x;\lambda)$ at terminal time $T$. The assumption \eqref{assump_Uiii} is made for $2 \leq i \leq 5$, but in fact it holds for the case $i = 1$ as a consequence stated in the following lemma.


\begin{lem}[\cite{KaZa:14}, Proposition 14]\label{lem_U}
Assume that the risk tolerance $R(x)$ satisfies: $R(0) = 0$,  $R(x)$ is strictly increasing and  $R'(x) < \infty$ on $[0,\infty)$, and there exists $K \in \RR^+$ such that
\begin{equation*}
\abs{\partial_x^{(2)} R^2(x)} \leq K,
\end{equation*}
then
\begin{equation}
R'(x) \leq C \quad \text{and} \quad R(x) \leq Cx, \quad \text{with} \quad C = \sqrt{K/2}.
\end{equation}
\end{lem}

\begin{lem}
The Asymptotic Elasticity condition \eqref{eq_usualconditions} is implied by the following condition:
\begin{equation*}
R(x) \leq Cx.
\end{equation*}
\end{lem}
\begin{proof}
Define the Arrow-Pratt risk aversion by:
\begin{equation}\label{eq_AP}
AP[U](x) = -x\frac{U''(x)}{U'(x)}.
\end{equation}
If follows directly from Proposition B.3 in \cite{SC:04}:
\begin{equation*}
\text{If } \liminf_{x\to +\infty} AP[U](x) = a > 0, \text{ then } AE[U] \leq (1-a)^+,
\end{equation*}
and consequently
$a = \liminf_{x\to+\infty} \left(\frac{x}{R(x)}\right) \geq \liminf_{x\to+\infty}\frac{x}{Cx} = \frac{1}{C} >0.$
\end{proof}

\begin{rem}\label{rem_U}
Assumption~\ref{assump_U}~\eqref{assump_Ubddbelow} is a sufficient assumption, in fact, there are cases where U(0+) is not finite, but our main Theorem \ref{Thm_one} still holds. For example, power utility $U(x) = \frac{x^\gamma}{\gamma}$ with $\gamma < 0$, and logarithmic utility $U(x) = \log(x)$. For the first case, the fully non-linear accuracy problem is completely solved in \cite{FoSiZa:13} by a distortion transformation, which linearizes the problem.

By expending $\partial_x^{(i)} R^i(x)$ in \eqref{assump_Uiii} and Lemma \ref{lem_U} , it is easily shown that Assumption \ref{assump_U} \eqref{assump_Urisktolerance}  is equivalent to the following conditions on the risk tolerance $R(x)$:  a) $R(0) = 0$, and $R(x)$ is strictly increasing on $[0, \infty)$; and b) $\abs{R^j(x)\left(\partial_x^{(j+1)}R(x)\right)} \leq K$, $\forall 0\leq j \leq 4$.
\end{rem}

\begin{prop}\label{prop_U}
The following classes of utility functions satisfy Assumption \ref{assump_U}:
\begin{enumerate}[(i)]
   \item Average of powers: $U(x) = \int_E x^y \nud{y}$, where $\nud{y}$ is a finite positive measure, and the support $E$ is compact, contained in $[0,1)$ and $\nu(\{0\}) = 0$. Two special cases are:
   \begin{enumerate}[a)]
   \item Power utility $U(x) = \frac{1}{\gamma}x^\gamma$, with $\gamma \in (0,1)$;
   \item Mixture of power utilities $U(x) = c_1\frac{x^{\gamma_1}}{\gamma_1} + c_2\frac{x^{\gamma_2}}{\gamma_2}$, with $\gamma_1, \gamma_2 \in (0,1)$ and  $c_1, c_2 >0$.
   \end{enumerate}
   In both cases, $\nud{y}$ is a counting measure of point(s) in [0,1).
   \item U(x) is given by positive inverse of the marginal utility $I(y)= U'^{(-1)}(y): \RR^+ \to \RR^+$,
   \begin{equation}
   I(y) = \int_0^N y^{-s}\nud{s},
   \end{equation}
   with $\nu$ being finite and positive on compact support $(N < +\infty)$. This is Example 18 in \cite{KaZa:14}, \done{and it satisfies condition \eqref{assump_Uiii} for $j=1,2$ as proved there.}
 \end{enumerate}
\end{prop}
The proof of Proposition \ref{prop_U} is left to Appendix~\ref{app_U}.

\begin{rem}
In the first class of utilities, $1 \notin E$ in general , unless further assumptions are prescribed on $\nud{y}$. For instance if $\nud{y}= \ud y$ and $E = [0,1]$, then, $AE[U] = \lim_{x\to+\infty}\frac{\ln(x)-1}{\ln(x)} =1$ which does not satisfy \eqref{eq_usualconditions}.
\end{rem}

\begin{rem}
For the power utility $U(x) = \frac{x^\gamma}{\gamma}$, the Arrow-Pratt risk aversion \eqref{eq_AP} is constant and the risk tolerance \eqref{eq_risktolerance} is linear, given by
\begin{equation*}
AP[U](x) = -x\frac{U''(x)}{U'(x)} = 1-\gamma, \qquad R(x) = \frac{x}{1-\gamma}.
\end{equation*}
However, general utilities, such as a mixture of two powers
$$U^{Mix}(x) = c_1\frac{x^{\gamma_1}}{\gamma_1} + c_2\frac{x^{\gamma_2}}{\gamma_2}, \quad 0<  \gamma_1 \leq \gamma_2 <1,$$
produce nonlinear risk aversion functions:
\begin{equation}\label{eq_apmixofpower}
AP[U^{Mix}](x) = \frac{c_1(1-\gamma_1)x^{\gamma_1 - \gamma_2} + c_2(1-\gamma_2)}{c_1x^{\gamma_1 - \gamma_2} + c_2},
\end{equation}
as well as nonlinear risk tolerances,
\begin{equation}
R(x) = \left(\frac{c_1x^{\gamma_1 - \gamma_2} + c_2}{c_1(1-\gamma_1)x^{\gamma_1-\gamma_2} + c_2(1-\gamma_2)}\right)x \sim \left\{
\begin{matrix}
\frac{x}{1-\gamma_2}, \text{ as } x \to \infty, \\
\frac{x}{1-\gamma_1}, \text{ as } x \to 0.
\end{matrix}
\right.
\end{equation}
This is illustrated in Figure \ref{fig-mixture}. Therefore, working with general utility enables us to model nonlinear relation between the relative risk aversion and the wealth (middle plot), and makes our model closer to results from empirical studies on how $AP[U](x)$ varies with wealth.
\begin{figure}[H]
  \centering
  \includegraphics[width=0.32\textwidth]{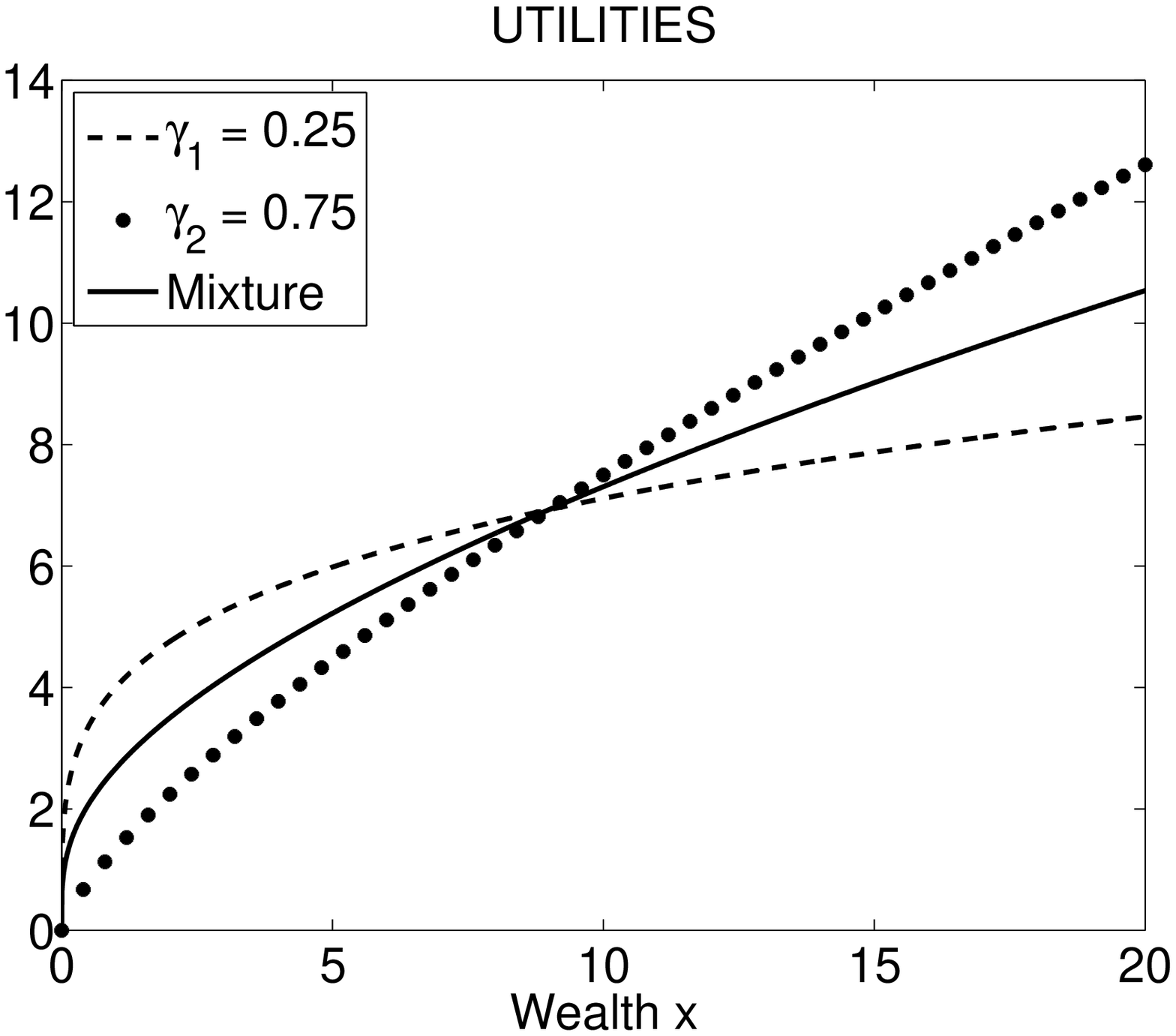}
  \includegraphics[width=0.32\textwidth]{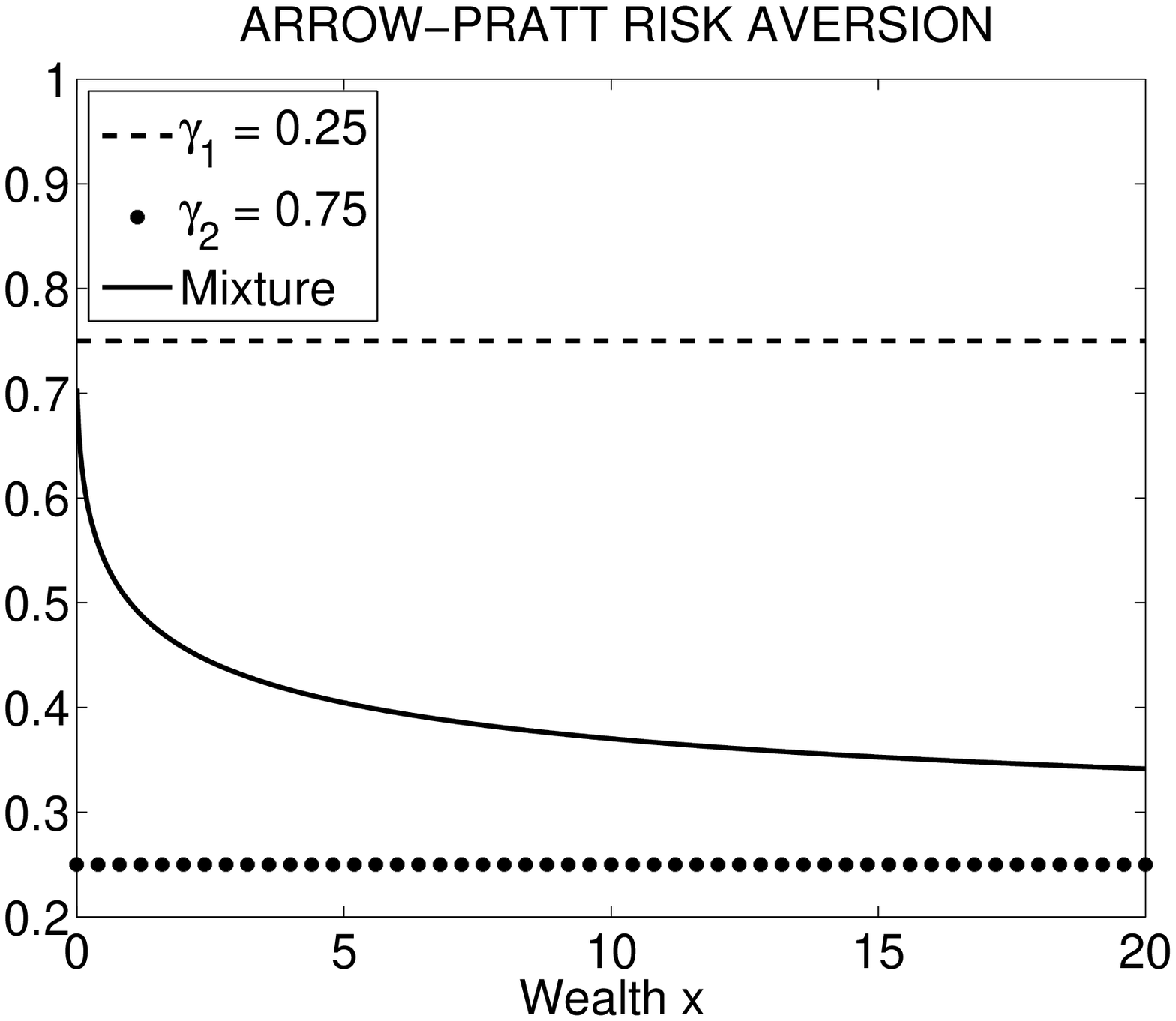}
  \includegraphics[width=0.32\textwidth]{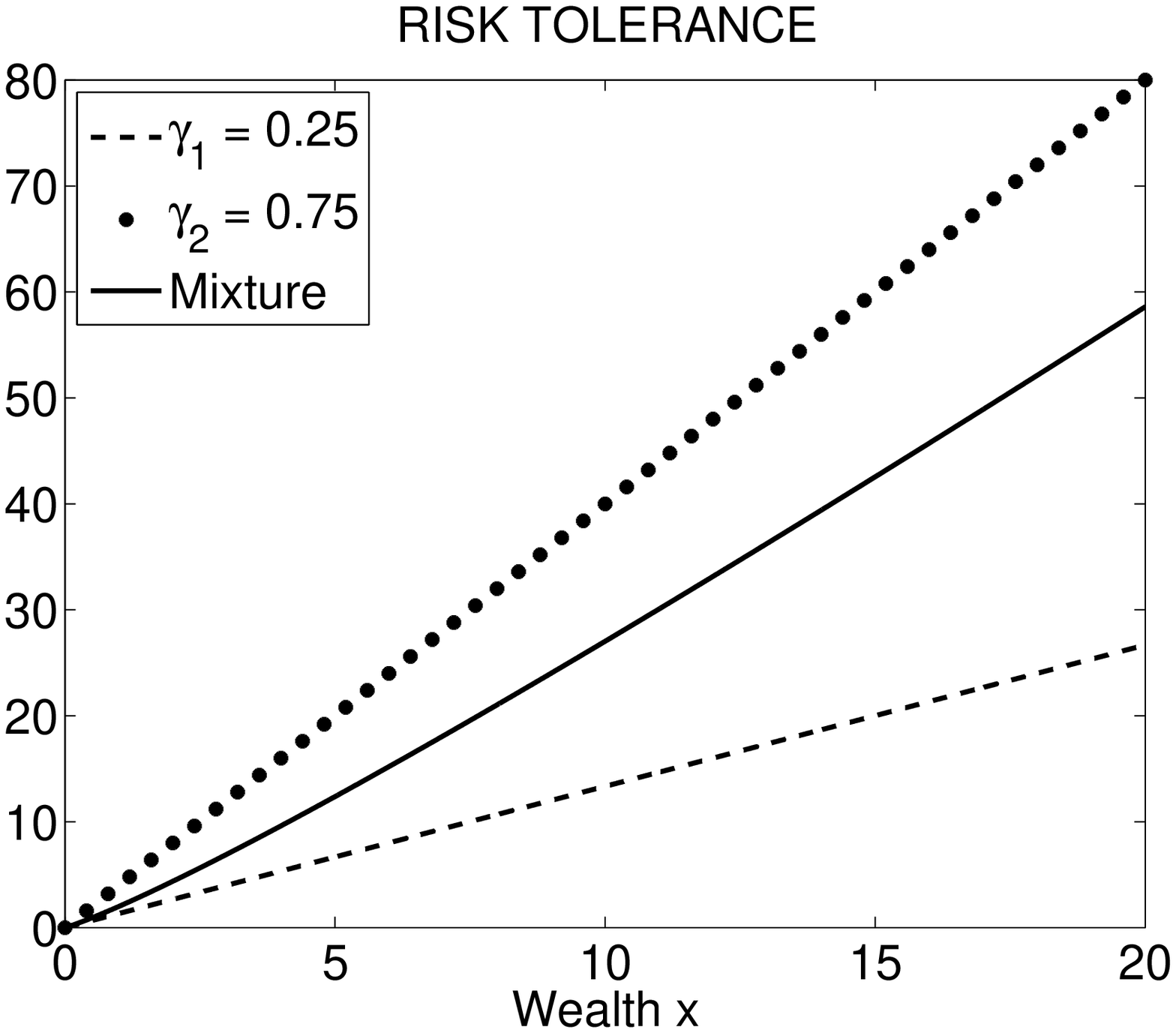}\\
  \caption{Mixture of power utilities with $\gamma_1 = 0.25$, $\gamma_2 = 0.75$ and  $c_1 = c_2 = 1/2$.}\label{fig-mixture}
\end{figure}

\end{rem}

\subsection{Assumptions on the State Processes $(X_t^\pz, S_t, Z_t)$}\label{sec_assumpvaluefunc}
Note that $z$ is only a parameter in the function $\vz(t,x,z) = M(t,x;\lambda(z))$ given by \eqref{eq_vz}, and for fixed $(t,z)$, $\vz$ is a concave function that has a linear upper bound. For $t=0$, there exists a function $G(z)$, so that $\vz(0,x,z) \leq G(z) + x, \forall (x,z) \in \RR^+\times \RR$.

The dynamics of the wealth process associated to the strategy $\pz(t,x,z) := \pi^\star(t,x;\lambda(z))$ given in \eqref{eq_pistar} and the slow factor $Z_t$ is given by:
\begin{align}\label{eq_Xt}
&\ud X_t^\pz = \pz(t,X_t^\pz,Z_t)\mu(Z_t)\ud t + \pz(t,X_t^\pz,Z_t)\sigma(Z_t)\ud W_t,\\
&\ud Z_t = \delta g(Z_t) \ud t  + \sqrt{\delta}c(z)\ud W_t^Z.\nonumber
\end{align}

\begin{assump}\label{assump_valuefunc} We make the following assumptions on the state processes $(X_t^\pz, S_t, Z_t)$:
\begin{enumerate}[(i)]
	\item\label{assump_valuefuncSZ} For any starting points $(s, z)$ and fixed $\delta$, the system of SDEs \eqref{eq_St}--\eqref{eq_Zt} has a unique strong solution $(S_t, Z_t)$. Moreover, the functions $\lambda(z)$ and $g(z)$ are in $C^3(\RR)$ and $C^2(\RR)$ respectively,
	 and the coefficients $g(z)$, $c(z)$, $\lambda(z)$ as well as their derivatives $g'(z)$, $g''(z)$, $\lambda'(z)$, $\lambda''(z)$ and $\lambda'''(z)$ are at most polynomially growing (see Remark \ref{remark:polynomiallygrowing}).
	
	\item\label{assump_valuefuncZmoment} The process $Z^{(1)}$ with infinitesimal generator $\MCM$ defined in \eqref{eq_Mgenerator} admits moments of any order uniformly in $t \leq  T$:
	\begin{equation}
	\sup_{t \leq T}\left\{ \EE\abs{Z_t^{(1)}}^k\right\} \leq C(T,k).
	\end{equation}
	    \item\label{assump_valuefuncG}  The process $G(Z_\cdot)$ is in $L^2([0,T]\times \Omega)$ uniformly in $\delta$, i.e.,
             \begin{equation}\label{assump_Gz}
             \EE_{(0,z)}\left[\int_0^T G^2(Z_s) \ud s\right] \leq C_1(T,z),
             \end{equation}
             where $C_1(T,z)$ is independent of $\delta$ and $Z_s$ follows \eqref{eq_Zt} with $Z_0 = z$.
         \item\label{assump_valuefuncX} The wealth process $X_\cdot^\pz$ is in  $L^2([0,T]\times \Omega)$ uniformly in $\delta$ , i.e.,
             \begin{equation}\label{assump_Xsquare}
             \EE_{(0,x,z)}\left[\int_0^T \left(X_s^\pz\right)^2 \ud s\right] \leq C_2(T,x,z),
             \end{equation}
             where $C_2(T,x,z)$ is independent of $\delta$ and $X_s^\pz$ follows \eqref{eq_Xt} with $X_0^\pz = x$.

       \end{enumerate}
\end{assump}

\begin{rem}\label{remark:polynomiallygrowing}
Note that in Assumption~\ref{assump_valuefunc}~\eqref{assump_valuefuncSZ},  the word ``polynomially growing'' is interpreted in different ways depending on the domain of $g(z), c(z)$ and $\lambda(z)$. For a function $h(z): \RR \to \RR$,  when $Z$ is an Ornstein--Uhlenbeck process for instance, then polynomial growth means that  there exists an integer $k$ and $a>0$, such that
\begin{equation*}
\abs{h(z)} \leq a (1+ \abs{z}^k).
\end{equation*}
Otherwise, if $h(z):\RR^+ \to \RR$, for example when $Z$ is a  Cox--Ingersoll-Ross process, then it means that there exists a $k \in \NN$ and $a>0$, such that
\begin{equation*}
\abs{h(z)} \leq a(1 + z^k + z^{-k}).
\end{equation*}

In Assumption~\ref{assump_valuefunc}~\eqref{assump_valuefuncG}, if the diffusion process $Z$ has exponential moments, then at-most exponential growth of $G(z)$ ensures  \eqref{assump_Gz}. An explicit example will be given in Section \ref{sec_example}.
\end{rem}

\begin{rem}
Note the Assumption \ref{assump_valuefunc} are on the utility function $U(x)$ through $\vz$, and the market model through $Z_\cdot$, but not on the unknown value function $\Vz$ itself. In fact, part of the assumption is to ensure the following estimate. 
\end{rem}
\begin{lem}\label{lem_unibdd}
Under Assumption~\ref{assump_valuefunc}~\eqref{assump_valuefuncG}-\eqref{assump_valuefuncX}, the process $\vz(\cdot, X_\cdot^\pz,Z_\cdot)$ is in $L^2([0,T]\times \Omega)$ uniformly in $\delta$, i.e. $\forall (t,x,z) \in [0,T]\times\RR^+\times \RR$:
\begin{equation}
\EE_{(t,x,z)}\left[\int_t^T \left(\vz(s,X_s^\pz, Z_s)\right)^2 \ud s\right] \leq C_3(T,x,z),
\end{equation}
where $\vz(t,x,z)$ is defined in Section \ref{sec_heuristic} and satisfies equation \eqref{eq_vz}.
\end{lem}
\begin{proof}
It follows by the straightforward computation:
\begin{align*}
\EE_{(t,x,z)}&\left[\int_t^T \left( \vz(s,X_s^\pz, Z_s)\right)^2 \ud s\right] \leq \EE_{(t,x,z)}\left[\int_t^T \left(\vz(0,X_s^\pz,Z_s)\right)^2 \ud s\right]\\
& \qquad \leq \EE_{(t,x,z)}\left[\int_t^T \left(G(Z_s) + X_s^\pz \right)^2 \ud s\right] \\
& \qquad \leq 2 \left( \EE_{(t,x,z)}\left[\int_t^T G^2(Z_s) \ud s\right] + \EE_{(t,x,z)}\left[\int_t^T \left(X_s^\pz \right)^2 \ud s\right] \right)\\
& \qquad \leq 2 \left( \EE_{(0,z)}\left[\int_0^T G^2(Z_s) \ud s\right] + \EE_{(0,x,z)}\left[\int_0^T \left(X_s^\pz \right)^2 \ud s\right] \right) \\
& \qquad = 2\left(C_1(T,z) + C_2(T,x,z) \right) := C_3(T,x,z),
\end{align*}
where we have successively used the monotonicity (decreasing property) of $\vz$ in $t$, the concavity of $\vz$ in $x$ and Assumptions~\ref{assump_valuefunc}~\eqref{assump_valuefuncG}-\eqref{assump_valuefuncX}.
\end{proof}

\section{First Order Approximation of the Value Function $\Vzl$}\label{sec_accuracy}

In this section, we analyze by perturbation methods the value function associated to the Merton strategy $\pz(t,x,z) = -\frac{\lambda(z)}{\sigma(z)}\frac{\vz_x}{\vz_{xx}}$. Assume $\pz$ is admissible, and recall that $X_t^\pz$ and $Z_t$ follows \eqref{eq_Xt} and \eqref{eq_Zt}, 
then, one defines the value function as the expected utility of terminal wealth:
\begin{equation}\label{def_Vzl}
\Vzl(t,x,z) = \EE\left\{U(X_T^\pz) | X_t^\pz = x, Z_t = z\right\},
\end{equation}
where $U(\cdot)$ is a general utility function satisfying Assumption \ref{assump_U}. The value function satisfies the linear PDE 
\begin{align}\label{eq_Vzl}
&\Vzl_t + \delta\MCM\Vzl + \frac{1}{2}\sigma^2(z)\left(\pz\right)^2\Vzl_{xx} + \pz\left(\mu(z)\Vzl_x + \sqrt{\delta}\rho  g(z)\sigma(z)\Vzl_{xz}\right) = 0, \\
&\Vzl(T,x,z) = U(x).\nonumber
\end{align}

Our main result of this section is:

\begin{theo}\label{Thm_one}
Under assumptions \ref{assump_U} and \ref{assump_valuefunc}, the residual function $E(t,x,z)$ defined by
$$E(t,x,z):= \Vzl(t,x,z) - \vz(t,x,z)-\sqrt\delta\vo(t,x,z),$$
where $\vz$ and $\vo$ are identified by \eqref{eq_vz} and \eqref{eq_vo}, is of order $\delta$.  In other words, $\forall (t,x,z) \in [0,T]\times\RR^+ \times \RR$, there exists a constant $C$, such that  $\abs{E(t,x,z)} \leq C\delta$, where $C$ may depend on $(t,x,z)$ but not on $\delta$.
\end{theo}

We recall that a function $f^\delta(t,x,z)$ is of order $\delta^k$, denoted by $f^\delta(t,x,z) \sim \MCO(\delta^k)$, if $\forall (t,x,z) \in [0,T]\times \RR^+\times \RR$, there exists $C$ such that $\abs{f^\delta(t,x,z)}\leq C\delta^k$, where $C$ may depned on $(t,x,z)$, but not on $\delta$. Similarly, we denote $f^\delta(t,x,z) \sim o(\delta^k)$, if $\limsup_{\delta\to0} |f^\delta(t,x,z)| / \delta^k = 0$.

\done{
\begin{cor}\label{cor_optimality}
In the case of power utility $U(x) = \frac{x^\gamma}{\gamma}$, $\pz$ is asymptotically optimal in $\MCA^\delta(t,x,z)$ up to order $\sqrt\delta$. 
\end{cor}
\begin{proof}
\cite[Corollary 6.8]{FoSiZa:13} proved that 
\begin{equation}
\Vz = \vz + \sqrt{\delta}\vo + \MCO(\delta),
\end{equation}
where $\Vz$ is defined in \eqref{def_Vz}. By Theorem \ref{Thm_one}, $\Vz$ and $\Vzl$ admit the same first order approximation $\vz + \sqrt{\delta}\vo$. Therefore we obtain that $\pz$ is asymptotically optimal in $\MCA^\delta(t,x,z)$ up to order $\sqrt\delta$. 

\end{proof}
}
\subsection{Estimate of Risk Tolerance Function\, $R(t,x;\lambda(z))$ and Leading Order Term\, $\vz$}

In this subsection, we state several properties of the risk tolerance function $R(t,x;\lambda(z))$, which will be needed in the proof of Theorem \ref{Thm_one}. Some of the proofs involve lengthy calculations which we put in the appendix.
\begin{prop}\label{prop_H}
Let $I: \RR^+ \to \RR^+$ be the inverse of marginal utility, and assume it satisfies the growth condition in Assumption \ref{assump_U} \eqref{assump_Ugrowth}. Also, define $H: \RR \times [0,T] \times \RR \to \RR^+$ by
\begin{equation}\label{eq_uandH}
M_x(t,H(x,t,\lambda(z));\lambda(z)) = \exp\{-x-\frac{1}{2}\lambda^2(z)(T-t)\},
\end{equation}
where $M(t,x;\lambda(z))$ is the Merton value function. Then:
\begin{enumerate}[(i)]
\item For each $\lambda(z)$, $H(x,t,\lambda(z))$ is the unique solution to the heat equation,
\begin{equation}\label{eq_H}
H_t + \frac{1}{2}\lambda^2(z)H_{xx} = 0,
\end{equation}
with the terminal condition $H(x,T,\lambda(z)) = I(e^{-x})$.

\item Moreover, for each $t \in [0,T]$ and $\lambda(z) \in \RR$, $H(x,t,\lambda(z))$ is strictly increasing and of full range,
\begin{equation}\label{eq_Hfullrange}
\lim_{x\to-\infty}H(x,t,\lambda(z)) = 0 \quad \text{and} \quad \lim_{x\to\infty} H(x,t,\lambda(z)) = \infty.
\end{equation}

\item Define the inverse function $H^{-1}(y,t,\lambda(z)): \RR^+ \times [0,T]\times \RR\to \RR$:
\begin{equation*}
H(H^{-1}(y,t,\lambda(z)),t,\lambda(z)) = y,
\end{equation*}
then, for $(t,x,z) \in [0,T]\times\RR^+\times\RR$, the risk tolerance function $R(t,x;\lambda(z))$ is given by
\begin{equation}\label{eq_rzH}
R(t,x;\lambda(z)) = H_x\left(H^{(-1)}(x,t,\lambda(z)),t,\lambda(z)\right).
\end{equation}
\end{enumerate}
\end{prop}
\done{
\begin{proof}
The	results under constant $\lambda$ with multiple assets are presented \cite[Propositions 4 and 6]{KaZa:14}. It is straightforward to generalize the results to $\lambda(z)$, as $z$ is a parameter. Therefore, here we omit the proof.
\end{proof}
}
%
%
%

\begin{prop}\label{prop_mono}
Suppose the risk tolerance $R(x) = -\frac{U'(x)}{U''(x)}$ is strictly increasing for all $x$ in $[0,\infty)$ (this is part of Assumption \ref{assump_U} \eqref{assump_Urisktolerance}), then, for each $t \in [0,T)$ and $\lambda(z) \in \RR$, the risk tolerance function $R(t,x;\lambda(z))$ is strictly increasing in the wealth variable $x$.
\end{prop}
\done{
\begin{proof}
We skip the proof by the same reasoning as in Proposition \ref{prop_H}, and refer to \cite[Proposition 9]{KaZa:14}.
\end{proof}
}
%

\begin{prop}\label{prop_ssh}
Under Assumption \ref{assump_U}, the risk tolerance function $R(t,x,\lambda(z))$ satisfies: $\forall \,0 \leq j \leq 4$, $\exists K_j >0$, such that $\forall (t,x,z) \in [0,T)\times \RR^+ \times \RR$,
\begin{equation}\label{eq_prop_Rbounds}
\abs{R^j(t,x;\lambda(z))\left(\partial_x^{(j+1)}R(t,x;\lambda(z))\right)} \leq K_j.
\end{equation}
Or equivalently, $\forall 1 \leq j \leq 5$, there exists $\widetilde K_j>0$, such that $\forall (t,x,z) \in [0,T)\times \RR^+ \times \RR$,
\begin{equation*}
\abs{\partial_x^{(j)} R^j(t,x;\lambda(z))} \leq \widetilde K_j.
\end{equation*}
Moreover, for $(t,x,z) \in [0,T)\times\RR^+\times \RR$,
\begin{align}\label{eq_prop_Rbound}
R(t,x;\lambda(z)) \leq K_0x.
\end{align}
\end{prop}
\done{\begin{proof}
The proof of Proposition~\ref{prop_ssh} for $0 \leq j \leq 4$ is given in Appendix \ref{app_ssh}. Note that results for $j=0,1$ follows from a generalization of \cite[Proposition 14]{KaZa:14}.
\end{proof}}

\begin{prop}\label{prop_rz}
The risk tolerance function $R(t,x;\lambda(z))$ satisfies the relation:
\begin{equation}\label{eq_rzgamma}
R_\lambda = (T-t)\lambda(z)R^2R_{xx}.
\end{equation}
\begin{proof}
Differentiating  \eqref{eq_vzgamma} with respect to $x$ gives:
$$
\vz_{xz} = (T-t)\lambda\lambda'(R_x\vz_x + R\vz_{xx})\quad \mbox{and}\quad
\vz_{xxz} = (T-t)\lambda\lambda'(R_{xx}\vz_x + 2R_x\vz_{xx} + R\vz_{xxx}).
$$
The definition \eqref{def_risktolerance} of $R(t,x;\lambda(z))$ and equation \eqref{eq_vzandmerton} imply that
$R_x = -1 + \frac{\vz_x\vz_{xxx}}{\left(\vz_{xx}\right)^2}$.
Differentiating \eqref{def_risktolerance} with respect to $z$, and using the above three equations produces
\begin{align*}
R_z &= \frac{-\vz_{xx}\vz_{xz} + \vz_x\vz_{xxz}}{\left(\vz_{xx}\right)^2} \\
& = (T-t)\lambda\lambda'\frac{-R_x\vz_x -R\vz_{xx}}{\vz_{xx}} + (T-t)\lambda\lambda'\frac{\vz_x}{\left(\vz_{xx}\right)^2}(R_{xx}\vz_x + 2R_x\vz_{xx} + R\vz_{xxx})\\
& = (T-t)\lambda\lambda'\left(R_xR-R+R^2R_{xx}-2R_xR+R (R_x+1)\right) = (T-t)\lambda\lambda'R^2R_{xx}.
\end{align*}
Then, the chain-rule relation $R_z = R_\lambda \lambda'(z)$ implies \eqref{eq_rzgamma}.
\end{proof}
\end{prop}

\begin{prop}\label{prop_vzzderivative}
Under Assumption \ref{assump_U} and Assumption \ref{assump_valuefunc}, there exist non-negative functions $d_{i,j}(z)$  and $\widetilde{d}_{i,j}(z)$ at most polynomially growing
such that the following inequalities are satisfied:
\begin{equation*}
\begin{array}{ll}
\abs{\vz_z(t,x,z)} \leq d_{01}(z)\vz(t,x,z),&\abs{\vz_{xz}(t,x,z)} \leq d_{11}(z)\vz_x(t,x,z), \\
\abs{\vz_{xxz}(t,x,z)} \leq d_{21}(z)\abs{\vz_{xx}(t,x,z)},&\abs{\rz_z(t,x;\lambda(z))} \leq \widetilde{d}_{01}(z)\rz(t,x;\lambda(z)),\\
\abs{\vz_{zz}(t,x,z)} \leq d_{02}(z)\vz(t,x,z),&\abs{\rz_{xz}(t,x;\lambda(z))} \leq \widetilde{d}_{11}(z),\\
\abs{\vz_{xzz}(t,x,z)} \leq d_{12}(z)\vz_x(t,x,z),&\abs{\rz_{zz}(t,x;\lambda(z))} \leq \widetilde{d}_{02}(z)\rz(t,x;\lambda(z)),\\
\abs{\vz_{xxzz}(t,x,z)} \leq d_{22}(z)\abs{\vz_{xx}(t,x,z)},&\abs{\vz_{xzzz}(t,x,z)} \leq d_{13}(z)\vz_x(t,x,z).
\end{array}
\end{equation*}
\end{prop}
\begin{proof}
The proof consists in successive differentiations starting with the ``Vega-Gamma'' relation in \eqref{eq_vzgamma}, and a repeated use of the concavity of $\vz$ and of the results in  Propositions \ref{prop_ssh} and  \ref{prop_rz}.
For the sake of space, we omit the details of this lengthy but straightforward derivation.

\end{proof}

\subsection{Proof of Theorem 3.1}
%

The heuristic expansion of $\Vzl$  is given by:
\begin{equation*}
\Vzl = \vz + \sqrt\delta\vo + \cdots.
\end{equation*}
and is derived in \cite[Appendix B]{FoSiZa:13}. Recall the residual function  $E(t,x,z)$ introduced in Theorem \ref{Thm_one}: 
$
E = \Vzl - \vz - \sqrt\delta\vo.
$
Subtracting \eqref{eq_vz} and \eqref{eq_vo} from \eqref{eq_Vzl}, one has
\begin{align}\label{eq_resone}
&E_t + \frac{1}{2}\sigma(z)^2\left(\pz\right)^2E_{xx} + \pz\mu(z)E_x + \delta\MCM E + \sqrt\delta\rho \sigma(z)g(z)\pz E_{xz}  \\
& \qquad + \delta\MCM(\vz + \sqrt\delta\vo) + \delta\rho \sigma(z)g(z)\pz \vo_{xz}  = 0, \quad E(T,x,z) = 0.\nonumber
\end{align}
Feynman--Kac formula gives the following probabilistic representation for $E(t,x,z)$
\begin{align*}
E(t,x,z) = &\delta\EE_{(t,x,z)}\bigg[\int_t^T \MCM\vz(s,X_s^\pz,Z_s) + \sqrt\delta\MCM\vo(s,X_s^\pz,Z_s) + \rho \sigma(Z_s)g(Z_s)\pz\vo_{xz}(s,X_s^\pz,Z_s) \ud s \bigg] \\
:= & \delta \text{I} + \delta^{3/2} \text{II} + \delta \rho \text{III},
\end{align*}
where $\EE_{(t,x,z)}[\cdot] = \EE[\cdot|X_t^\pz = x, Z_t = z]$ and
\begin{align}
& \text{I} :=\EE_{(t,x,z)}\left[\int_t^T c(Z_s) \vz_{z}(s,X_s^\pz,Z_s) + \frac{1}{2}g^2(Z_s)\vz_{zz}(s,X_s^\pz,Z_s)
 \ud s \right], \label{eq_I} \\
& \text{II}:= \EE_{(t,x,z)}\left[\int_t^T c(Z_s) \vo_{z}(s,X_s^\pz,Z_s) + \frac{1}{2}g^2(Z_s)\vo_{zz}(s,X_s^\pz,Z_s)
 \ud s \right], \label{eq_II}\\
& \text{III} := \EE_{(t,x,z)}\left[\int_t^T
\lambda(Z_s)g(Z_s)\rz(s,X_s^\pz; \lambda(Z_s))\vo_{xz}(s,X_s^\pz,Z_s)\ud s \right].\label{eq_III}
\end{align}
In order to show that $E$ is of order $\delta$, it suffices to show that I, II and III are uniformly bounded in $\delta$. 

We first analyze term I in \eqref{eq_I}. The boundedness for the  $z$-derivatives of $\vz$ is given by Proposition \ref{prop_vzzderivative}. To bound the $L^2$ norm of $\vz(\cdot, X_\cdot^\pz, Z_\cdot)$ we rely on Lemma \ref{lem_unibdd}. In the following we omit the arguments  of $\vz(s,X_s^\pz,Z_s)$ and its derivatives.
\begin{align*}
\text{I} &= \EE_{(t,x,z)}\left[\int_t^T
c(Z_s)\vz_z + \frac{1}{2}g^2(Z_s)\vz_{zz}\ud s \right] \doteq \;\text{I}^{(1)} + \frac{1}{2} \;\text{I}^{(2)}.
\end{align*}
\begin{align*}
\abs{\text{I}^{(1)}} &\leq \EE_{(t,x,z)}\left[\int_t^T
\abs{c(Z_s)\vz_z}\ud s \right] \leq \EE_{(t,x,z)}\left[\int_t^T \abs{c(Z_s)d_{01}(Z_s)}\vz\ud s \right] \\
&\leq \EE_{(t,z)}^{1/2}\left[\int_t^T c^2(Z_s)d^2_{01}(Z_s) \ud s\right]\EE_{(t,x,z)}^{1/2}\left[\int_t^T \left(\vz\right)^2\ud s\right]\\
& \leq C(T,z) C_3(T,x,z).
\end{align*}
In the calculation above, $\vz_z$ is replaced by its bound $d_{01}(z)\vz$ derived in Proposition \ref{prop_vzzderivative}. By Cauchy-Schwarz inequality, it suffices to bound two expectations. For the first one, we have used the facts that $c(z)$ and $d_{01}(z)$ have at most polynomial growth and $Z_t$ admits moments of any order uniformly in $\delta$. Lemma \ref{lem_unibdd} gives the bound for the second expectation. 

The bounds of remaining terms are obtained by the same procedure.
\begin{align*}
\abs{\text{I}^{(2)}} &\leq \EE_{(t,x,z)}\left[\int_t^T
g^2(Z_s)\abs{\vz_{zz}}\ud s \right] \leq \EE_{(t,x,z)}\left[\int_t^T g^2(Z_s)d_{02}(Z_s)\vz\ud s \right] \\
&\leq \EE_{(t,z)}^{1/2}\left[\int_t^T g^4(Z_s)d^2_{02}(Z_s) \ud s\right]\EE_{(t,x,z)}^{1/2}\left[\int_t^T \left(\vz\right)^2\ud s\right]\\
& \leq C(T,z) C_3(T,x,z).
\end{align*}

Term II in \eqref{eq_II} and term III in \eqref{eq_III} contain derivatives in $z$ of $\vo$. To deal with it, we recall the following relation between $\vo$ and $\vz$ given by equation \eqref{eq_vzandvo}:
\begin{align*}
\vo = -\frac{1}{2}(T-t)\rho \lambda(z)g(z)\frac{\vz_{x}\vz_{xz}}{\vz_{xx}} = \frac{1}{2} (T-t)\rho \lambda(z)g(z)\rz\vz_{xz}.
\end{align*}
Differentiating the above equation with respect to $z$, we are able to rewrite $\vo_{xz}$, $\vo_z$ and $\vo_{zz}$ in terms of the risk tolerance function $\rz$ and the leading order term $\vz$. Then, as before, the derivations are mainly based on Proposition \ref{prop_vzzderivative} and Lemma \ref{lem_unibdd}, and we omit the details here.

\section{Asymptotic Optimality of $\pz$}\label{sec_optimality}
%

The goal of this section is to show that the strategy $\pz$ defined in \eqref{pi0},  asymptotically outperforms every family $\MCA_0(t,x,z)\left[\pzt, \pot,\alpha\right]$ defined (\ref{def_A0}), as precisely stated in our main Theorem \ref{Thm_optimality} in Section~\ref{sec_intro}.

Denote by $\Vzt$ the value function associated to the trading strategy $\pi := \pzt + \delta^\alpha \pot \in \MCA_0(t,x,z)\left[ \pzt, \pot, \alpha\right]$:
\begin{equation}\label{def_Vzt}
\Vzt = \EE\left[U(X_T^\pi) \vert X_t^\pi = x, Z_t = z \right],
\end{equation}
where $X_t^\pi$ is the wealth process following the strategy $\pi$, and $Z_t$ is slowly varying  with the same $\delta$:
\begin{align}\label{def_Xtopt}
&\ud X_t^\pi = \pi(t,X_t^\pi,Z_t)\mu(Z_t)\ud t + \pi(t,X_t^\pi,Z_t)\sigma(Z_t)\ud W_t,\\
&\ud Z_t = \delta c(Z_t)\ud t + \sqrt\delta g(Z_t)\ud W_t^Z. \label{def_Ztopt}
\end{align}
We need to compare $\Vzt$ with $\Vzl$ defined in  \eqref{def_Vzl}, for which we have established the first order approximation $\vz+\sqrt{\delta}\vo$ in Theorem \ref{Thm_one}. This comparison is asymptotic in $\delta$ up to order $\sqrt{\delta}$, and our first step is to obtain the corresponding approximation for $\Vzt$. This is done heuristically in Section \ref{sec_heuristicVzt} in the two cases $\pzt \equiv \pz$ and $\pzt \not\equiv \pz$, and depending on the value of the parameter $\alpha$.  The proof of accuracy is given in Section \ref{sec_accuracy2}. Asymptotic optimality of $\pz$ is obtained in Section \ref{sec_optimalitypz}.

\begin{assump}\label{assump_piregularity}
	For a fixed choice of $(\pzt$, $\pot$, $\alpha>0)$, we require:
	\begin{enumerate}[(i)] 
		\item The whole family (in $\delta$) of strategies $\{\pzt + \delta^\alpha \pot\}$ is contained in $\MCA^\delta(t,x,z)$;
		\item  Functions $\pzt(t,x,z)$ and $\pot(t,x,z)$ are continuous on $[0,T]\times \RR^+\times \RR$;
		\item Let  $(\widetilde{X}_s^{t,x})_{t\leq s\leq T}$ be the solution to:
		\begin{equation}\label{eq_Xttilde}
		\ud \widetilde X_s = \mu(z)\pzt(s,\widetilde X_s, z) \ud s + \sigma(z) \pzt(s,\widetilde X_s,z) \ud W_s,
		\end{equation}
		starting at $x$ at time $t$. 
		
		By (i), $\widetilde{X}_s^{t,x}$ is nonnegative and we further
		assume that it has full support $\RR^+$ for any $t<s\leq T$.
	\end{enumerate}
\end{assump}
\begin{rem}\label{rem_pztpot}
	Notice that $\pz$ defined in \eqref{pi0} is continuous on $[0,T]\times \RR^+\times \RR$, thus, it is natural to  require that $\pzt$ and $\pot$ have the same regularity as $\pz$, that is (ii). Regarding (iii),
	from  Section \ref{sec:assumptions}, $\pz$ is the optimal trading strategy for the Merton problem when $\delta =0$, in which case $Z_t$ is  frozen at its initial position $z$. The associated wealth process $\widehat X_s^{t,x}$  starting at $x$ at time $t$ is the solution to
	\begin{equation*}
	\ud \widehat X_s = \mu(z)\pz(s,\widehat X_s, z) \ud s + \sigma(z)\pz(s,\widehat X_s, z) \ud W_s, \quad \widehat X_t = x.
	\end{equation*}
	Then, from \cite[Proposition 7]{KaZa:14}, one has	
	\begin{equation*}
	\widehat X_s^{t,x} = H\left(H^{-1}(x,t,\lambda(z)) + \lambda^2(z)(s-t) + \lambda(z)(W_s-W_t), s, \lambda(z)\right),
	\end{equation*}
	where $H: \RR\times[0,T]\times \RR \to \RR^+$ is defined in Proposition \ref{prop_H} and is of full range. Consequently,  $\widehat X_s^{t,x}$ has full support $\RR^+$, and thus, it is natural to require that $\widetilde X_s^{t,x}$ has full support $\RR^+$, that is (iii). 
\end{rem}
\begin{rem}
	We have $\MCA_0(t,x,z)\left[\pzt, \pot, 0\right]=\MCA_0(t,x,z)\left[\pzt + \pot, 0, \alpha\right]$, so that it is enough to consider $\alpha >0$.
\end{rem}

\subsection{Heuristic Expansion of the Value Function $\Vzt$}\label{sec_heuristicVzt}
We look for an expansion of the value function $\Vzt$ defined in \eqref{def_Vzt} of the form
\begin{equation}\label{eq_Vztexpansion}
\Vzt = \vzt + \delta^\alpha \vat + \delta^{2\alpha}\vtat + \cdots + \delta^{n\alpha}\widetilde v^{n\alpha} + \sqrt{\delta}\,\vot +\cdots,
\end{equation}
where $n$ is the largest integer such that $n\alpha < 1/2$. Note that in the case $\alpha > 1/2$, $n$ is simply zero. In the derivation, we are interested in identifying the zeroth order term $\vzt$ and the first non-zero term up to order $\sqrt\delta$. The term following $\vzt$ will depend on the value of $\alpha$. 

Denote by $\MCL$ the infinitesimal generator of the state processes $(X_t^\pi,Z_t)$ given by \eqref{def_Xtopt} - \eqref{def_Ztopt}
\begin{equation*}
\MCL :=  \delta\MCM + \frac{1}{2}\sigma^2(z)\left(\pzt + \delta^\alpha \pot \right)^2\partial_{xx} + \left(\pzt + \delta^\alpha\pot\right)\mu(z)\partial_x + \sqrt{\delta}\rho  g(z)\sigma(z)\left(\pzt + \delta^\alpha\pot\right)\partial_{xz},
\end{equation*}	
then, the value function $\Vzt$ defined in \eqref{def_Vzt} satisfies
\begin{equation}\label{eq_Vzt}
\partial_t \Vzt + \MCL \Vzt =0, \qquad \Vzt(T,x,z) = U(x).
\end{equation}

Collecting terms of order one yields the equation satisfied by $\vzt$
\begin{align}\label{eq_vzt}
&\vzt_t + \frac{1}{2}\sigma^2(z)\left(\pzt\right)^2\vzt_{xx} + \mu(z)\pzt\vzt_x = 0, \\
&\vzt(T,x,z) = U(x).\nonumber 
\end{align}
The order of approximation will depend on $\pzt$ being identical to $\pz$ or not.

\subsubsection{Case $\pzt \equiv \pz$}\label{sec_pzteqpz}
In this case, from the definition \eqref{pi0} of $\pz$, equation \eqref{eq_vzt} becomes \eqref{def_ltxpde} which is also satisfied by $\vz$ by \eqref{eq_vzandmerton}. By Proposition \ref{prop_ltxunique}, we deduce $\vzt \equiv \vz$.
To identify the term of next order, one needs to discuss case by case:
\begin{enumerate}[(i)]
	\item 
	$\alpha = 1/2$.
	The next order term is $\vot$ and it satisfies
	\begin{align*}
	&\vot_t + \frac{1}{2}\sigma^2\left(\pz\right)^2\vot_{xx} + \pz\mu(z)\vot_x  + \pz\rho  g(z)\sigma(z)\vzt_{xz} + \pot\left(\sigma^2(z)\pz\vzt_{xx}+\mu(z)\vzt_x\right) = 0,\\
	&\vot(T,x,z) = 0.\nonumber
	\end{align*}
	It reduces to equation \eqref{eq_vo}
	since we have the relations
	\begin{equation}\label{eq_vzlrelation}
	\vzt = \vz \quad \text{ and }\quad \sigma^2(z)\pz\vzt_{xx} = -\mu(z)\vzt_x,
	\end{equation}
	from the definition \eqref{pi0} of $\pz$. From  Section \ref{sec_heuristic} item (ii), $\vo$ is the unique solution to \eqref{eq_vo} and therefore, we obtain $\vot \equiv \vo$.
	
	\item 
	$\alpha > 1/2$. The next order is of $\MCO(\delta^{1/2})$.  By collecting all terms of order $\delta^{1/2}$, we also obtain that  $\vot$ satisfies \eqref{eq_vo}, and $\vot \equiv \vo$.
	
	\item 
	$\alpha < 1/2$. The next order correction is $\MCO(\delta^\alpha)$. Collecting all terms of order $\delta^{\alpha}$ in \eqref{eq_Vzt} yields
	\begin{align}\label{eq_vat}
	&\vat_t + \frac{1}{2}\sigma^2(z)\left(\pz\right)^2\vat_{xx} + \pz\mu(z)\vat_x + \pot\left(\sigma^2(z)\pz\vzt_{xx} + \mu(z)\vzt_x \right)= 0,\\
	&\vat(T,x,z) = 0.\nonumber
	\end{align}
	The last two terms cancel via the relation \eqref{eq_vzlrelation}, and \eqref{eq_vat} becomes \eqref{eq_vo} with $\rho=0$, which only has the trivial solution, namely $\vat \equiv 0$. Therefore, we need to identify the next non-vanishing term.
	\begin{itemize}
		\item 
		$1/4 < \alpha < 1/2$.
		The next order is of $\MCO(\delta^{1/2})$, and $\vot$ satisfies
		$$
		 \vot_t + \frac{1}{2}\sigma^2(z)\left(\pz\right)^2\vot_{xx} + \pz\mu(z)\vot_x +\rho \pz g(z)\sigma(z)\vzt_{xz} = 0, \quad
		 \vot(T,x,z) = 0.
		$$
		It coincides with \eqref{eq_vo} and we deduce $\vot = \vo$. 
		
		\item 
		$\alpha = 1/4$.
		The next order is of $\MCO(\delta^{1/2})$, and the PDE satisfied by $\vot$ becomes
		\begin{align}\label{eq_vot}
		&\vot_t + \frac{1}{2}\sigma^2(z)\left(\pz\right)^2\vot_{xx} + \pz\mu(z)\vot_x + \frac{1}{2}\sigma^2(z)\left(\pot\right)^2\vzt_{xx} + \pz\rho g(z)\sigma(z)\vzt_{xz} = 0,\\
		& \vot(T,x,z) = 0,\nonumber
		\end{align}
		which will be used later when we compare $\vo$ and $\vot$.
		\item 
		$0 < \alpha < 1/4$.
		The next order is of $\MCO(\delta^{2\alpha})$ since $2\alpha < 1/2$, and
		\begin{equation}\label{eq_vtat}
		\vtat_t + \frac{1}{2}\sigma^2(z)\left(\pz\right)^2 \vtat_{xx} + \pz\mu(z)\vtat_{x} + \frac{1}{2}\sigma^2(z)\left(\pot\right)^2\vzt_{xx} =0,\quad
		\vtat(T,x,z) = 0.
		\end{equation}
		Feynman--Kac formula gives: 
		\begin{align}\label{feynman_vtat}
		\vtat(t,x,z) = \EE\left[\int_t^T \frac{1}{2}\sigma^2(z)\left(\pot\right)^2(s,\widetilde X_s,z)\vzt_{xx}(s,\widetilde X_s,z) \ud s\big\vert \widetilde X_t = x\right],
		\end{align}
		with $\widetilde X_s$ following \eqref{eq_Xttilde}.
		Notice that, for  fixed $z$, if the source term $\frac{1}{2}\sigma^2(z)\left(\pot\right)^2\vzt_{xx}$ is identically zero after some time $t_1$, then, $\vtat(t_1,x,z)$ is zero. Therefore, further analysis is needed in order to find the first non-zero term after $\vzt$ at point $(t,x,z)$. Note that both $\sigma(z)$ and $(-\vzt_{xx})$ are strictly positive ( $\vzt = \vz$ is strictly concave), hence, $\pot$ is the problematic term. Accordingly, we define
				\begin{equation*}
		t_1(z) = \inf\{t\in [0,T]: \pot(u,x,z) = 0, \forall (u,x)\in [t,T]\times \RR^+\},
		\end{equation*}
		where we use the convention $\inf \{\emptyset\} =T$.
		Based on $t_1(z)$, the following two regions are defined: 
		\begin{align}\label{def_regionk1}
		\MCK_1 &= \left\{(t,x,z): 0\leq t < t_1(z), x \in \RR^+, z \in \RR\right\}, \\
		\MCC_1 &= \left\{(t,x,z): t_1(z)\leq t\leq T, x \in \RR^+, z \in \RR\right\},\label{def_regionc1}
		\end{align}
		which form a partition of $[0,T]\times\RR^+\times \RR$.
		\begin{itemize}
			\item 
			For any $(t,x, z)\in \MCK_1$, since $t< t_1(z)$, there exists a point $(t',x', z)\in [t,t_1(z))\times \RR^+\times\{z\}$ such that 
			$\pot(t',x',z)\neq 0$. By  continuity of $\pot$, there exist $\eta >0$ and a set $A:=[t',t'+\epsilon]\times [x',x'+\epsilon]$ with 
			$0<\epsilon <t_1(z)-t'$ such that  $|\pot|\geq \eta$  on $A\times\{z\}$. By \eqref{feynman_vtat} and denoting by $\mu_s$
			the distribution of $\widetilde X_s^{t,x}$,
			we deduce that
			\begin{align}
			\vtat(t,x,z)&\leq \frac{1}{2}\sigma^2(z)\int_{x'}^{x'+\epsilon}\int_{t'}^{t'+\epsilon}\left(\pot\right)^2(s,y,z)\vzt_{xx}(s,y,z) \ud s \,\mu_s(\ud y)\nonumber\\
			&\leq -\frac{1}{2}\sigma^2(z)\eta^2\int_{x'}^{x'+\epsilon}\int_{t'}^{t'+\epsilon}[-\vzt_{xx}(s,y,z)] \ud s \,\mu_s(\ud y)\nonumber\\
			&\leq -\frac{1}{2}\sigma^2(z)\eta^2 \inf_A \left[-\vzt_{xx}(s,y,z)\right]  \int_{t'}^{t'+\epsilon}\left(\int_{x'}^{x'+\epsilon}\mu_s(\ud y)\right)    \ud s \nonumber\\
			&<0.\label{eq_vtat_strictineq}
			\end{align}
			The conclusion $\vtat(t,x,z)<0$ follows from $\vzt\equiv \vz$, strict concavity and continuity of $\vz$, and  the full-support assumption on  the distribution $\mu_s$ of $\widetilde X_s^{t,x}$.
			\item 
			For any $(t,x, z)\in \MCC_1$, equation \eqref{eq_vtat} becomes \eqref{eq_vo} with $\rho =0$ (since $\pot \equiv 0$ in $\MCC_1$), and consequently, $\vtat(t,x,z) \equiv 0$. Therefore, we need to analyze the next order term. Recall that  $n$ is the largest integer such that $n\alpha < 1/2$ and we are in the case $0<\alpha <1/4$. 
			\begin{itemize}
				\item			
				If $n=2$, collecting terms of order $\delta^{1/2}$ and using the facts that $\vtat \equiv 0 
				$ in $\MCC_1$ and $\vat \equiv 0$, yields \eqref{eq_vo} for $\vot$, and therefore,  $\vot = \vo$.
				\item
				For $n \geq 3$, namely, the next order is $\delta^{3\alpha}$ and $\alpha < 1/6$, then $\vthat$ satisfies
				\begin{align}
				&\vthat_{t} + \frac{1}{2}\sigma^2(z)\left(\pz\right)^2\vthat_{xx} + \pz\mu(z)\vthat_x + \sigma^2(z)\pz\pot\vtat_{xx} + \mu(z)\pot\vtat_{x}=0, \\
				&\vthat(T,x,z) = 0.\nonumber
				\end{align}
				Notice that in the above PDE, $z$ is simply a parameter. For fixed $z$, on the region $[t_1(z),T]\times \RR^+$, $\vtat(t,x,z) \equiv 0$ and the above equation  reduces to \eqref{eq_vo} with $\rho = 0$ again. Therefore, $\vthat(t,x,z) \equiv 0$ in the region $\MCC_1$. 
				Repeating this argument until $\widetilde v^{n\alpha}$, we obtain  
				\begin{equation*}
				\widetilde v^{i\delta}(t,x,z) \equiv 0, \quad 2 \leq  i \leq n, \quad  \forall (t,x,z) \in \MCC_1,
				\end{equation*}
				and, as in the case $n=2$, we conclude $\vot = \vo$.
			\end{itemize}
		\end{itemize}
	\end{itemize}
\end{enumerate}	
We summarize the above discussion in the following table:
{\small \begin{table}[H]
	\centering
	\caption{Expansion of $\Vzt$ when $\pzt \equiv \pz$.}\vspace{5pt}\label{tab_eq}
	\begin{tabular}{c|c|c}
		\hline\hline
		Value of $\alpha$&Expansion&Remark \\ \hline
		$\alpha \geq 1/2$ & $\vz + \sqrt{\delta}\vo$& \\ \cline{1-1}
		$1/4 < \alpha < 1/2$& &   \\ \hline 
		$\alpha = 1/4$& $\vz + \sqrt{\delta}\vot$& $\vot$ satisfies equation \eqref{eq_vot} \\\hline
		& Region $\MCK_1$: $\vz + \delta^{2\alpha}\vtat$& $\vtat$ satisfies equation \eqref{eq_vtat} and \eqref{eq_vtat_strictineq} \\ 
		$0 < \alpha <1/4$	& Region $\MCC_1$: $\vz + \sqrt{\delta}\vo$&  \\ \hline\hline
	\end{tabular}
\end{table}
}

\subsubsection{Case $\pzt \not\equiv \pz$}
Recall that the leading order term $\vzt$ satisfies \eqref{eq_vzt}:
\begin{align*}
\vzt_t + \frac{1}{2}\sigma^2(z)\left(\pzt\right)^2\vzt_{xx} + \pzt\mu(z)\vzt_x = 0, \quad \vzt(T,x,z) = U(x).
\end{align*}
For  $z\in \RR$, we introduce
\begin{equation*}
t_0(z) = \inf\left\{t\geq 0: \pzt(u,x,z) \equiv \pz(u,x,z), \forall (u,x) \in [t,T]\times\RR^+ \right\}, \quad \inf\{\emptyset\}=T.
\end{equation*}
Define the regions: 
\begin{align}\label{def_regionk}
\MCK &= \{(t,x,z): 0 \leq t < t_0(z), x \in \RR^+, z\in \RR\}, \\
\MCC &= \{(t,x,z): t_0(z) \leq t \leq T, x \in \RR^+, z \in \RR\}.\label{def_regionc}
\end{align}
We claim that in the region $\MCK$, $\vzt$ and $\vz$ differ, while in the region $\MCC$, $\vzt \equiv \vz$ and we need to identify the next non-varnishing term.

In order to compare $\vz$ and $\vzt$, we rewrite the equation \eqref{eq_vz} satisfied by $\vz$ as:
\begin{equation*}
\vz_t + \frac{1}{2}\sigma^2(z)\left(\pzt\right)^2\vz_{xx} + \pzt\mu(z)\vz_x - \frac{1}{2}\sigma^2(z)\left(\pzt - \pz\right)^2\vz_{xx}  = 0,
\end{equation*}
where we have used the relation
$
-\sigma^2(z)\pz\vz_{xx} =\mu(z)\vz_x
$.

Now let $f(t,x,z) := \vz(t,x,z) - \vzt(t,x,z)$ be the difference of the two leading order terms,
it satisfies
$$
f_t + \frac{1}{2}\sigma^2(z)\left(\pzt\right)^2 f_{xx} + \pzt\mu(z)f_x - \frac{1}{2}\sigma^2(z)\left(\pzt - \pz\right)^2\vz_{xx} = 0, \quad
f(T,x,z) = 0.
$$
By the Feymann-Kac formula, one has:
\begin{equation}\label{feymann_vzt}
f(t,x,z) = -\EE\left[\int_t^T \frac{1}{2}\sigma^2(z)\left(\pzt - \pz\right)^2(s, \widetilde X_s, z)\vz_{xx}(s,\widetilde X_s,z)\ud s \Big\vert \widetilde X_t = x\right],
\end{equation}
where $\widetilde X_s$ follows \eqref{eq_Xttilde}.
Using the argument given  in Section \ref{sec_pzteqpz} for the case  $0 < \alpha < 1/4$, we deduce that 
the right-hand side in \eqref{feymann_vzt} is strictly positive. Consequently $f(t,x,z) >0$, and
\begin{equation}\label{strictineq}
\vzt(t,x,z) < \vz(t,x,z), \quad  \forall (t,x,z) \in \MCK.
\end{equation}
Thus, in that case, the next term will not play a role when comparing $\Vzt$ and $\Vzl=\vz+\sqrt{\delta}\vo+\MCO(\delta)$.

For any $(t,x,z) \in \MCC$, since  we have $\pzt \equiv \pz$ on $\MCC$, we can apply here the whole discussion in Section \ref{sec_pzteqpz} (on the partition $\{\MCC\cap \MCK_1,\MCC \cap \MCC_1\}$ in the case $0 < \alpha < 1/4$).	
The expansion results are summarized in the table:	
{\small \begin{table}[H]
	\centering
	\caption{Expansion of $\Vzt$ when $\pzt \not\equiv \pz$.}\vspace{5pt}\label{tab_neq}
	\begin{tabular}{c|c|c|c}
		\hline\hline
		Region&Value of $\alpha$&Expansion&Remark \\ \hline
		$\MCK$& all& $\vzt$& $\vzt$ satisfies  \eqref{eq_vzt} and  \eqref{strictineq}  \\\hline
		&$\alpha \geq 1/2$ & $\vz + \sqrt{\delta}\vo$& \\ \cline{2-2}
		$\MCC$&$1/4 < \alpha < 1/2$& &   \\  \cline{2-4}
		&$\alpha = 1/4$& $\vz + \sqrt{\delta}\vot$& $\vot$ satisfies equation \eqref{eq_vot} \\\hline
		$\MCC\cap \MCK_1$&& $\vz + \delta^{2\alpha}\vtat$& $\vtat$ satisfies equation \eqref{eq_vtat} and \eqref{eq_vtat_strictineq} \\ 
		$\MCC \cap \MCC_1$&$0 < \alpha <1/4$& $\vz + \sqrt{\delta}\vo$&  \\ \hline\hline
	\end{tabular}
\end{table}
}
\subsection{Accuracy of Approximations}\label{sec_accuracy2}
\begin{prop}\label{prop_piaccuracy}
	Under Assumptions \ref{assump_U} \eqref{assump_Uregularity}-\eqref{assump_Ubddbelow},  \ref{assump_piregularity} and \ref{assump_optimality}, we obtain the following accuracy results:
{\small	\begin{table}[H]
		\centering
		\caption{Accuray of approximations of $\Vzt$.}\vspace{5pt}\label{tab_accuracy}
		\begin{tabular}{c|c|c|c|c}
			\hline\hline
			Case&Region&Value of $\alpha$&Approximation&Accuracy\\ \hline
			&&$\alpha \geq 1/2$ & $\vz + \sqrt{\delta}\vo$& $\MCO(\delta)$ \\ \cline{3-3}\cline{5-5}
			&all&$1/4 < \alpha < 1/2$& & $\MCO(\delta^{2\alpha})$   \\ \cline{3-5} 
			$\pzt \equiv \pz$&&$\alpha = 1/4$& $\vz + \sqrt{\delta}\vot$& $\MCO(\delta^{3/4})$ \\\cline{2-5}
			& $\MCK_1$&& $\vz + \delta^{2\alpha}\vtat$&$\MCO(\delta^{3\alpha \wedge (1/2)})$ \\ 
			&$\MCC_1$&$0 < \alpha <1/4$& $\vz + \sqrt{\delta}\vo$&$ \MCO(\delta)$ \\ \hline
			
			&$\MCK$& all& $\vzt$&$ \MCO(\delta^{\alpha \wedge (1/2)})$ \\\cline{2-5}
			&&$\alpha \geq 1/2$ & $\vz + \sqrt{\delta}\vo$&$ \MCO(\delta)$ \\ \cline{3-3}\cline{5-5}
			$\pzt \not\equiv \pz$&$\MCC$&$1/4 < \alpha < 1/2$& &$ \MCO(\delta^{2\alpha})$    \\  \cline{3-5}
			&&$\alpha = 1/4$& $\vz + \sqrt{\delta}\vot$&$ \MCO(\delta^{3/4})$ \\\cline{2-5}
			&$\MCC\cap \MCK_1$&& $\vz + \delta^{2\alpha}\vtat$&$ \MCO(\delta^{3\alpha \wedge (1/2)})$ \\ 
			&$\MCC \cap \MCC_1$&$0 < \alpha <1/4$& $\vz + \sqrt{\delta}\vo$&$  \MCO(\delta)$  \\ \hline\hline
		\end{tabular}
	\end{table}
	}
	\noindent where  the meaning of  $\MCO$ is as in Theorem \ref{Thm_one}. 
	\end{prop}
	In order to make rigorous the above expansions, we need  additional assumptions listed in Appendix \ref{appendix_addasump}. They are technical integrability conditions, uniformly in $\delta$, on the strategies in the class $\MCA_0(t,x,z)\left[\pzt, \pot,\alpha\right]$ defined in \eqref{def_A0} and their associated wealth processes.

\begin{proof}
	Recall that $\Vzt$ satisfies
	\begin{align}\label{eq_Vztfull}
	&\Vzt_t + \delta \MCM\Vzt + \frac{1}{2}\sigma^2(z)\left(\pzt + \delta^\alpha \pot \right)^2\Vzt_{xx} + \left(\pzt + \delta^\alpha\pot\right)\mu(z)\Vzt_x + \sqrt{\delta}\rho g(z)\sigma(z)\left(\pzt + \delta^\alpha\pot\right)\Vzt_{xz}=0,\\
	&\Vzt(T,x,z) = U(x).\nonumber
	\end{align}
	
	The  proofs of accuracy for the approximations given in Tables \ref{tab_eq} and  \ref{tab_neq} are quite standard, and we sketch them following the order of Table \ref{tab_accuracy}. 
	In each case, $E$ denotes the difference between $\Vzt$ and its approximation. It satisfies the terminal condition $E(T,x,z) = 0$ which we do not repeat below.
	
	We start with the case $\pzt = \pz$.
	\begin{enumerate}[(i)]
		\item 
		$\alpha = 1/2$. Subtracting equation \eqref{eq_vz} and \eqref{eq_vo} from \eqref{eq_Vztfull}, we obtain the PDE satisfied by $E(t,x,z)$:
		\begin{align*}
		&E_t + \MCL E + \delta \MCM(\vz + \sqrt\delta{\vo}) + \frac{\delta}{2}\sigma^2\left(\pot\right)^2\left(\vz_{xx} + \sqrt\delta\vo_{xx}\right) + \delta \sigma^2(z)\pz\pot\vo_{xx} \\
		& \qquad + \delta\mu(z)\pot\vo_x + \delta\rho g(z)\sigma(z)\left(\pz\vo_{xz} + \pot\vz_{xz} + \sqrt\delta\pot\vo_{xz}\right) = 0.
		\end{align*}			
		
	    Then, Feynman--Kac formula produces
		\begin{align*}
		E(t,x,z) &= \delta\EE_{(t,x,z)}\int_t^T\left[  \MCM \vz(s,X_s^\pi,Z_s) + \frac{1}{2}\sigma^2(Z_s)\left(\pot\right)^2 \vz_{xx}(s,X_s^\pi,Z_s) \right]\ud s \\
		&\qquad  + \delta^{3/2} \EE_{(t,x,z)}\int_t^T\left[  \MCM \vo(s,X_s^\pi,Z_s) + \frac{1}{2}\sigma^2(Z_s)\left(\pot\right)^2\vo_{xx}(s,X_s^\pi,Z_s) \right]\ud s   \\
		&\qquad + \delta\EE_{(t,x,z)}\int_t^T \left[\sigma^2(Z_s)\pz\pot\vo_{xx}(s,X_s^\pi,Z_s) + \mu(Z_s)\pot\vo_x(s,X_s^\pi,Z_s)\right]  \ud s \\
		&\qquad  + \delta\rho \EE_{(t,x,z)}\int_t^T\left[ g(Z_s)\sigma(Z_s)\pz\vo_{xz}(s,X_s^\pi,Z_s) + g(Z_s)\sigma(Z_s)\pot\vz_{xz}(s,X_s^\pi,Z_s) \right] \ud s \\
		&\qquad  + \delta^{3/2}\rho  \EE_{(t,x,z)}\int_t^T g(Z_s)\sigma(Z_s)\pot\vo_{xz}(s,X_s^\pi,Z_s) \ud s .
		\end{align*}
		
		Under Assumption \ref{assump_optimality} \eqref{assump_optimality_eqmore}, one has $E= \MCO(\delta)$.
		\item
		$\alpha > 1/2$. Similarly, we have
		\begin{align*}
		&E_t + \MCL E + \delta \MCM(\vz + \sqrt\delta{\vo}) + \frac{\delta^{2\alpha}}{2}\sigma(z)^2\left(\pot\right)^2\left(\vz_{xx} + \sqrt\delta\vo_{xx}\right) + \delta^{1/2+\alpha} \sigma^2(z)\pz\pot\vo_{xx} \\
		& \qquad + \delta^{1/2+\alpha}\mu(z)\pot\vo_x + \delta\rho g(z)\sigma(z)\left(\pz\vo_{xz} + \delta^{\alpha - 1/2}\pot\vz_{xz} + \delta^{\alpha}\pot\vo_{xz}\right) = 0.
		\end{align*}
		By Feynman--Kac formula and  Assumption \ref{assump_optimality} \eqref{assump_optimality_eqmore}, we deduce $E= \MCO(\delta)$.

		\item 
		$1/4< \alpha < 1/2$. We have
		\begin{align*}
		&E_t + \MCL E + \delta \MCM(\vz + \sqrt\delta{\vo}) + \frac{\delta^{2\alpha}}{2}\sigma(z)^2\left(\pot\right)^2\left(\vz_{xx} + \sqrt\delta\vo_{xx}\right) + \delta^{1/2+\alpha} \sigma^2(z)\pz\pot\vo_{xx} \\
		& \qquad + \delta^{1/2+\alpha}\mu(z)\pot\vo_x + \delta^{1/2+\alpha}\rho g(z)\sigma(z)\left(\delta^{1/2-\alpha}\pz\vo_{xz} + \pot\vz_{xz} + \sqrt\delta\pot\vo_{xz}\right) = 0,
		\end{align*}	
		and
		by Assumption \ref{assump_optimality} \eqref{assump_optimality_eqmore},  we have
		$E= \MCO(\delta^{2\alpha})$.
		
		\item 
		$\alpha = 1/4$. Subtracting equation \eqref{eq_vz} and \eqref{eq_vot} from \eqref{eq_Vztfull} yield
		\begin{align*}
		&E_t + \MCL E + \delta \MCM(\vz + \sqrt\delta{\vot}) + \frac{\delta^{3/4}}{2}\sigma(z)^2\left(2\pz\pot + \delta^{1/4}(\pot)^2\right)\vot_{xx} + \delta^{3/4}\mu(z)\pot\vot_x \\
		& \qquad  + \delta^{3/4}\rho g(z)\sigma(z)\left(\pot\vz_{xz} + \sqrt{\delta}\pot\vot_{xz} + \delta^{1/4}\pz\vot_{xz}\right) = 0,
		\end{align*}
		and
		Assumption \ref{assump_optimality} \eqref{assump_optimality_eqequal} implies $E= \MCO(\delta^{3/4})$.

		\item 
		$0 < \alpha < 1/4$. In the region $\MCK_1$, subtracting \eqref{eq_vz} and \eqref{eq_vtat} from \eqref{eq_Vztfull} produces
		\begin{align*}
		&E_t + \MCL E + \delta \MCM(\vz + \delta^{2\alpha}{\vtat}) + \frac{\delta^{3\alpha}}{2}\sigma(z)^2\left(2\pz\pot + \delta^\alpha(\pot)^2\right)\vtat_{xx} + \delta^{3\alpha}\mu(z)\pot\vtat_x \\
		& \qquad  + \sqrt\delta\rho g(z)\sigma(z)\left(\pz + \delta^\alpha\pot\right)\left(\vz_{xz} + \delta^{2\alpha}\vtat_{xz}\right) = 0,
		\end{align*}	
		and
		by Assumption \ref{assump_optimality} \eqref{assump_optimality_eqless}, one concludes that 
		$E=\MCO(\delta^{3\alpha \wedge (1/2)})$.

		In the complementary region $\MCC_1$, $E$ satisfies:
		\begin{align*}
		&E_t + \MCL E + \delta \MCM(\vz + \sqrt\delta{\vo}) + \frac{\delta^{2\alpha}}{2}\sigma(z)^2\left(\pot\right)^2\left(\vz_{xx} + \sqrt\delta\vo_{xx}\right) + \delta^{1/2+\alpha} \sigma^2(z)\pz\pot\vo_{xx} \\
		& \qquad + \delta^{1/2+\alpha}\mu(z)\pot\vo_x + \sqrt\delta\rho g(z)\sigma(z)\left(\sqrt{\delta}\pz\vo_{xz} + \delta^{\alpha}\pot\vz_{xz} + \delta^{\alpha+1/2}\pot\vo_{xz}\right) = 0.
		\end{align*}
		Note that in the region $\MCC_1$, $\pot \equiv 0$, and the above equation reduces to:
		\begin{align*}
		&E_t + \MCL E + \delta \MCM(\vz + \sqrt\delta{\vo}) +  \delta\rho g(z)\sigma(z)\pz\vo_{xz} = 0,		\end{align*}
				and then, Assumption \ref{assump_optimality} \eqref{assump_optimality_eqless} implies $E=\MCO(\delta)$.
	\end{enumerate}
	
	Now, we turn to  the case $\pzt \not\equiv \pz$. In the region $\MCK$, we know by
	\eqref{strictineq} that $\vzt<\vz$. Therefore, $\Vzt-\vz$ is asymptotically of order one and negative. Thus, the next term will not play a role
	and we define $E= \Vzt - \vzt$. 
	Subtracting equation \eqref{eq_vzt} from \eqref{eq_Vztfull} gives
	\begin{align*}
	&E_t  + \MCL E +  \delta\MCM \vzt + \sqrt\delta\left(\pzt + \delta^\alpha \pot\right) \rho g(z)\sigma(z)\vzt_{xz} + \frac{1}{2}\sigma^2(z)\left(\delta^\alpha\pot\right)^2\vzt_{xx} \\
	&\hspace{200pt}+ \sigma^2(z)\pzt\delta^\alpha\pot\vzt_{xx} + \delta^\alpha\pot\mu(z)\vzt_x =0.
	\end{align*}
	By Assumption \ref{assump_optimality} \eqref{assump_optimality_neq}, we  conclude that
	$E= \MCO(\delta^{\alpha \wedge (1/2)})$.
	
	Remark that $\pzt \equiv \pz$  in the region $\MCC$. Therefore, the whole analysis of case $\pzt \equiv \pz$ can be applied here, except that the case $0 < \alpha < 1/4$, where the accuracy results hold in the partition $\{\MCC \cap \MCK_1, \MCC \cap \MCC_1\}$ instead of $\MCC$. This complete the proof.	
\end{proof}

\subsection{Asymptotic Optimality: Proof of Theorem \ref{Thm_optimality}}\label{sec_optimalitypz}

The main result in this section is the proof of Theorem \ref{Thm_optimality}.
%
%
%

In order to compare the asymptotic performance of $\pz$ with the family of trading strategies \\
$\MCA_0(t,x,z)\left[ \pzt, \pot,\alpha\right]$, we are essentially comparing the approximations of $\Vzt$ summarized in Table \ref{tab_accuracy} with the first order approximation $\vz + \sqrt{\delta}\vo$ of $\Vzl$ obtained in Theorem \ref{Thm_one}. In each case in Table \ref{tab_accuracy} where the approximation of $\Vzt$ is $\vz + \sqrt{\delta}\vo$, it is easy to check that \eqref{eq_Vztineq} is satisfied and the limit is zero.
	The remaining five cases are: 
	(a) $\pzt \equiv \pz$ and $\alpha = 1/4$; 
	(a') $\pzt \not\equiv \pz$, in $\MCC$, and $\alpha = 1/4$;
	(b) $\pzt \equiv \pz$ and $0 < \alpha<1/4$ in the region $\MCK_1$;
	(b') $\pzt \not\equiv \pz$, in $\MCC\cap\MCK_1$, and $0 < \alpha<1/4$;  
	and (c) $\pzt \not\equiv \pz$ in the region $\MCK$.
	\begin{enumerate}[(a)]
		\item
		In the case $\pzt \equiv \pz$ and $\alpha  = 1/4$, the approximation of $\Vzt$ up to order $\sqrt{\delta}$ is $\vz + \sqrt\delta \vot $, and it suffices to show that $\vot \leq \vo$ for all $(t,x,z) \in [0,T]\times\RR^+\times \RR$. 
		Subtracting \eqref{eq_vot} from \eqref{eq_vo} shows that the difference $f(t,x,z) := \vo(t,x,z) - \vot(t,x,z)$ satisfies
		$$
		f_t + \frac{1}{2}\sigma^2(z)\left(\pz\right)^2 f_{xx} + \pz\mu(z)f_x - \frac{1}{2}\sigma^2(z)\left(\pot\right)^2\vz_{xx} = 0, \quad
		f(T,x,z) = 0,
		$$
		and  admits the representation
		\begin{equation*}
		f(t,x,z) = -\EE\left[\int_t^T \frac{1}{2}\sigma^2(z)\left(\pot\right)^2(s,\widetilde X_s,z)\vz_{xx}(s,\widetilde X_s,z)\ud s \Big\vert \widetilde X_t = x\right],
		\end{equation*}
		where $\widetilde X_t$ follows \eqref{eq_Xttilde}. 
		The concavity of $\vz$ implies $f(t,x,z)\geq 0$ and therefore, \eqref{eq_Vztineq} holds.
		\item
		In the case $\pzt \equiv \pz$ and $0 < \alpha < 1/4$, the approximation of $\Vzt$ is $\vz + \delta^{2\alpha}\vtat + o(\delta^{3\alpha\wedge 1/2})$, where $\vtat$ is strictly negative by \eqref{eq_vtat_strictineq}. Consequently, 
		\begin{equation*}
		\lim_{\delta\to 0}\frac{\Vzt(t, x, z)-\Vzl(t, x, z)}{\sqrt{\delta}}=\lim_{\delta\to 0}\frac{\delta^{2\alpha}\vtat-\sqrt{\delta}\vo+\MCO(\delta^{3\alpha\wedge 1/2})}{\sqrt{\delta}}=-\infty,
		\end{equation*}
		and \eqref{eq_Vztineq} holds.
		\item
		In the case $\pzt \not\equiv \pz$ and $ (t,x,z) \in \MCK$, the approximation of $\Vzt$ is $\vzt + o(1)$, and \eqref{strictineq} shows that $\vzt$ is strictly less than $\vz$. Thus, we deduce \eqref{eq_Vztineq}.
	\end{enumerate}
	The proof for the case (a') ({resp.} (b')) is essentially the same as in (a) (resp. (b)) but in the region $\MCC$ (resp. $\MCC\cap\MCK_1$).

\section{A Fully-Solvable Example}\label{sec_example}

In this section, we consider a model 
studied in \cite{ChVi:05} where explicit solutions are derived for the consumption problem over infinite horizon, and in \cite{FoSiZa:13} where expansions for the terminal wealth problem are derived and accuracy of approximation is proved under power utility with one factor. Our goal is to show that this model satisfies the various assumptions we have made in this paper and, therefore, justify that they are reasonable.
The underlying asset $S_t$ and the slowly varying factor $Z_t$ are modeled by:
\begin{align}\label{def_Stexample}
&\ud S_t = \mu S_t \ud t + \sqrt{\frac{1}{Z_t}} S_t\ud W_t, \\
&\ud Z_t = \delta (m-Z_t)\ud t + \sqrt\delta \beta \sqrt{Z_t} \ud W_t^Z,\label{def_Ztexample}
\end{align}
with $\beta>0$ and $\mu >0$.
The standard Feller condition $\beta^2 \leq 2m$  is assumed  to ensure that $Z_t$ stays positive. In this example, we consider power utilities:
$$U(x) = \frac{x^\gamma}{\gamma}, \quad 0 < \gamma < 1,$$
for which Assumption \ref{assump_U} is satisfied by Proposition \ref{prop_U}.
This model fits in the class of models \eqref{eq_St}-\eqref{eq_Zt} 
by identifying the coefficients $\mu(z)$, $\sigma(z)$, $c(z)$ and $g(z)$ as follows:
\begin{equation*}
\mu(z) = \mu, \quad \sigma(z) = \sqrt{1/z}, \quad c(z) = m-z, \quad g(z) = \beta\sqrt z.
\end{equation*}

For Assumption~\ref{assump_valuefunc}~\eqref{assump_valuefuncSZ}-\eqref{assump_valuefuncZmoment}  (with state space $(0,\infty)$), we notice that $(Z_t)$ is the unique strong solution to \eqref{def_Ztexample} and it has finite moments of any order uniformly in $\delta\leq 1$ and $t \leq T$, see for instance \cite[Chapter 3]{FoPaSiSo:11}.
The process $(S_t)$ is given by:
\begin{equation*}
S_t = S_0\exp\left( \int_0^t \left(\mu - \frac{1}{2Z_s}\right)\ud s+ \int_0^t \sqrt{\frac{1}{Z_s}} \ud W_s\right).
\end{equation*}
For Assumption~\ref{assump_valuefunc}~\eqref{assump_valuefuncG}-\eqref{assump_valuefuncX}, we first solve \eqref{eq_vz} to obtain $\vz$ and $\pz$: 
\begin{equation*}
\vz_t - \frac{1}{2}\mu^2 z \frac{\left(\vz_x\right)^2}{\vz_{xx}} = 0,\qquad
\vz(T,x,z) = \frac{x^\gamma}{\gamma}.
\end{equation*}
One can easily check that 
\begin{equation*}
\vz(t,x,z) = \frac{x^\gamma}{\gamma} e^{\frac{\mu^2\gamma}{2(1-\gamma)}z(T-t)},
\end{equation*}
is a solution, and by Proposition \ref{prop_ltxunique} it is the unique solution.
Consequently, the zeroth order strategy $\pz$ and the risk tolerance function $R(t,x;\lambda(z))$ are given by
\begin{equation}\label{eq_pzexample}
\pz(t,x,z) = \frac{\mu xz}{1-\gamma}, \qquad \text{and} \qquad R(t,x;\lambda(z)) = \frac{x}{1-\gamma}.
\end{equation}
Note that in this case, the relations on the derivatives of $\vz$ in Proposition \ref{prop_vzzderivative} can be verified by direct computation.
The verification of Assumption \ref{assump_valuefunc} \eqref{assump_valuefuncG}-\eqref{assump_valuefuncX} will be presented in the next two sections.

\subsection{Integrability of the Process $G(Z_\cdot)$}\label{sec_veriassumpii}
As in \cite{AnPi:07}, one can compute the left-hand sides of \eqref{assump_Gz} and \eqref{assump_Xsquare} by solving Riccati equations. 
As mentioned in Section \ref{sec_assumpvaluefunc}, $\vz(0,x,z)$ is a concave function, and it has a linear upper bound $G(z) + x$. To obtain $G(z)$, we derive: $\forall x_0 \in \RR^+$,
\begin{align*}
\vz(0,x,z) &\leq \vz(0,x_0,z) + \frac{\partial}{\partial x}\vz(0,x_0,z)(x-x_0) \\
& = \left(\frac{1}{\gamma}-1\right)x_0^\gamma e^{\frac{\mu^2\gamma z}{2(1-\gamma)}T} + x_0^{\gamma-1}e^{\frac{\mu^2\gamma z}{2(1-\gamma)}T}x.
\end{align*}
Let $x_0 = e^{\frac{\mu^2\gamma z}{2(1-\gamma)^2}T}$ so that the coefficient in front of $x$ is 1, and $G(z)$ can be chosen as:
\begin{align*}
G(z) = \left(\frac{1}{\gamma}-1\right)x_0^\gamma e^{\frac{\mu^2\gamma z}{2(1-\gamma)}T} = \left(\frac{1}{\gamma}-1\right)e^{\frac{\mu^2\gamma z}{2(1-\gamma)^2}T}.
\end{align*}
We have
\begin{equation}\label{momentsG}
\EE_{(0,z)}\left[\int_0^T G^2(Z_s) \ud s\right] = \left(\frac{1}{\gamma}-1\right)^2\int_0^T f^\delta(0,z;s)\ud s,
\end{equation}
where
\begin{equation*}
f^\delta(t,z;s) = \EE\left[e^{\frac{\mu^2\gamma T}{(1-\gamma)^2}Z_s}\bigg\vert Z_t = z\right],
\end{equation*}
solves
\begin{align}
&f^\delta_t + \frac{\delta}{2}\beta^2 zf^\delta_{zz} + \delta(m-z)f^\delta_z = 0, \quad t \in [0,s),\label{eq_f}\\
&f^\delta(s,z;s) = e^{wz},\quad \mbox{with}\quad w = \frac{\mu^2\gamma T}{(1-\gamma)^2}.\nonumber
\end{align}
This equation admits the solution
\begin{equation}\label{fexplicit}
f^\delta(t,z;s) = e^{wz + A^\delta(s-t)z + B^\delta(s-t)},
\end{equation}
where $A^\delta(\tau)$ satisfies the Riccati equation:
\begin{align}
&A^\delta(\tau)' = \frac{\delta}{2}\beta^2A^\delta(\tau)^2 + \left(\delta\beta^2 w-\delta\right)A^\delta(\tau) + \left(\frac{\delta}{2}\beta^2w^2-\delta w\right), \quad \tau \in (0,s],\\
&A^\delta(0) = 0,
\end{align}
and $B^\delta(\tau)$ solves
\begin{equation}\label{eq_B}
B^\delta(\tau)' = \delta m (w + A^\delta(\tau)), \quad B^\delta(0) = 0.
\end{equation}
The discriminant of this equation is $\Delta = \delta^2$ which is positive, and  one gets: 
\begin{equation}\label{eq_A}
A^\delta(\tau) = \frac{-w\left(1-e^{-\delta \tau}\right)}{1-\frac{w}{w-\frac{2}{\beta^2}}e^{-\delta\tau}}, \quad \tau \in [0,\tau^\star(\delta)),
\end{equation}
where $[0,\tau^\star(\delta))$ is the domain where $A^\delta(\tau)$ stays finite. 
It remains to show that $A^\delta(\tau)$ and $B^\delta(\tau)$ are uniformly bounded in $(\delta,\tau)\in [0,\overline\delta]\times[0,T]$ for some $\overline\delta\leq 1$. 
Note that the boundedness of $B^\delta(\tau)$ is a consequence of that of $A^\delta(\tau)$ via equation \eqref{eq_B}. Since $A^\delta(\tau)$ is continuous on $(0,1]\times [0,\tau^\star(\delta))$,  it suffices to show that i) there exists $\overline\delta$, such that $\tau^\star(\delta)>T$ for $\delta \leq \overline\delta$, and ii) $\lim_{\delta \to 0} A^\delta(\tau)$ exists. To this end, we examine the following cases:
\begin{itemize}
	\item[(a)] $w < \frac{2}{\beta^2}$. The denominator  of \eqref{eq_A} stays above 1,   $\tau^\star(\delta) = \infty$, and  $\lim_{\delta\to 0}A^\delta(\tau)=0$.
	\item[(b)] $w > \frac{2}{\beta^2}$. Here $\tau^\star(\delta) = -\frac{1}{\delta}\ln\left(\frac{w-\frac{2}{\beta^2}}{w}\right)$, $\lim_{\delta \to 0}\tau^\star(\delta) = \infty$, and $\lim_{\delta\to 0}A^\delta(\tau)=0$.
\item[(c)]
$w=\frac{2}{\beta^2}$. This case gives the trivial solution $A^\delta(\tau) \equiv 0$.
\end{itemize}
In all cases, $A^\delta(\tau)$ is uniformly bounded in $[0,\overline\delta]\times[0,T]$
and therefore, combined with \eqref{momentsG} and \eqref{fexplicit}, we deduce that Assumption \ref{assump_valuefunc} \eqref{assump_valuefuncG} is satisfied. 
%
%

\subsection{Moments of the Wealth Process $X_t^\pz$}\label{sec_wealth}

First, using the explicit formula \eqref{eq_pzexample} for $\pz$,  equation \eqref{eq_Xt} becomes
\begin{equation}
\ud X_s^\pz = \frac{\mu^2Z_s}{1-\gamma} X_s^\pz \ud s + \frac{\mu\sqrt Z_s}{1-\gamma}X_s^\pz \ud W_s, \quad s \geq t.
\end{equation}
In order to control $\EE_{(0,x,z)}\left[\int_0^T \left(X_s^\pz\right)^2 \ud s \right]$, we introduce
$
f^\delta(t,x,z;s) := \EE\left[\left(X_s^\pz\right)^2\bigg \vert X_t^\pz = x, Z_t = z\right],
$
which solves
\begin{align}
&f^\delta_t + \frac{\mu^2z}{1-\gamma}xf^\delta_x + \frac{1}{2}\frac{\mu^2z}{(1-\gamma)^2}x^2f^\delta_{xx} + \delta(m-z)f^\delta_z + \frac{\delta}{2}\beta^2zf^\delta_{zz} + \rho \frac{\sqrt\delta\mu\beta}{1-\gamma}zxf^\delta_{xz} = 0,\\
&f^\delta(s,x,z;s) = x^2.
\end{align}
The solution is of the form
$
f^\delta(t,x,z;s) = x^2e^{A^\delta(s-t)z + B^\delta(s-t)},
$
where $A^\delta(\tau)$ satisfies the Riccati equation:
\begin{align}\label{eq_A'}
A^\delta(\tau)' = \frac{\delta}{2}\beta^2A^\delta(\tau)^2 + \left(\frac{2\sqrt\delta\rho \mu\beta}{1-\gamma}-\delta\right)A^\delta(\tau) + \frac{(3-2\gamma)\mu^2}{(1-\gamma)^2}, \quad \tau \in (0,s], \quad
A^\delta(0) = 0,
\end{align}
and $B^\delta(\tau)$ solves
\begin{equation}\label{eq_B'}
B^\delta(\tau)' = \delta m A^\delta(\tau), \quad B^\delta(0) = 0.
\end{equation}
By a similar argument used in Section \ref{sec_veriassumpii}, the verification of the uniform bound
\begin{align*}
\EE_{(0,x,z)}\left[\int_0^T X_s^2 \ud s \right]=\int_0^Tf^\delta(0,x,z;s)\ud s  \leq C_2(T,x,z),
\end{align*}
reduces to i) there exists $\overline\delta$, such that $\tau^\star(\delta)>T$ for $\delta \leq \overline\delta$, (recall that $\tau^\star(\delta)$ is defined to be the  explosion time) and ii) $\lim_{\delta \to 0} A^\delta(\tau)$ exists. We omit the details which are lengthy but straightforward analysis.
\done{
\subsection{Asymptotic Optimality of $\pz$}
So far, we have shown that all assumptions listed in Section~\ref{sec:assumptions} are satisfied. Thus, under the model \eqref{def_Stexample}--\eqref{def_Ztexample} and power utility, Theorem \ref{Thm_one} is valid and we have $\Vzl = \vz + \sqrt{\delta}\vo + \MCO(\delta)$. In this section, an example with specific $\pzt$ and $\pot$ is presented, and we verify the Assumption \ref{assump_piregularity} and \ref{assump_optimality} under such choice. Consequently, following Theorem \ref{Thm_optimality}, $\pz$ outperforms this class of strategies up to order $\sqrt{\delta}$.

We choose  
\begin{equation*}
 \pzt = \pot = \pz = \frac{\mu xz}{1-\gamma},
\end{equation*}
then $X_t^\pi$ defined in \eqref{def_Xtopt} follows
\begin{equation}
\ud X_t^\pi = (1+\delta^\alpha)\frac{\mu^2 Z_t}{1-\gamma}X_t^\pi \ud t + (1+\delta^\alpha)\frac{\mu \sqrt{Z_t}}{1-\gamma}X_t^\pi \ud W_t.
\end{equation}
Assumption \ref{assump_piregularity} is easily verified: (i) solving the above SDE yields
\begin{equation*}
X_t^\pi = \exp\left\{\int_0^t (1+\delta^\alpha)\frac{\mu^2Z_s}{1-\gamma} - \frac{1}{2}(1+\delta^\alpha)^2\frac{\mu^2Z_s}{(1-\gamma)^2}\ud s  + \int_0^t (1+\delta^\alpha) \frac{\mu\sqrt{Z_s}}{1-\gamma}\ud W_s\right\},
\end{equation*} 
which ensures the admissibility of the family of strategy $\{\pzt + \delta^\alpha\pot\}_{\delta>0}$; (ii) the function $\frac{\mu xz}{1-\gamma}$ is continuous by definition; and (iii) the full support property of $\widetilde X_s^{t,x}$, for any $t < s \leq T$, follows from $\pzt = \pz$ and Remark \ref{rem_pztpot}.

We now check the Assumption \ref{assump_optimality} case by case: 

(a) $\alpha > 1/4$. Recall $\vz$ from \eqref{eq_pzexample} and obtain $\vo$ by \eqref{eq_vzandvo}:
\begin{equation*}
\vz(t,x,z) = \frac{x^\gamma}{\gamma}e^{\frac{\mu^2\gamma}{2(1-\gamma)}z(T-t)}, \quad \vo(t,x,z) = \frac{\gamma x^\gamma}{4(1-\gamma)^2}(T-t)^2\rho\mu^3\beta z e^{\frac{\mu^2\gamma}{2(1-\gamma)}z(T-t)},
\end{equation*}
we deduce that all quantities required to be uniformly bounded in $\delta$ in Assumption \ref{assump_optimality} are of form
\begin{equation}\label{eq_exform}
\EE_{(t,x,z)}\int_t^T \mc{P}(Z_s) \left(X_s^\pi\right)^\gamma \ud s,
\end{equation}
where $\mc{P}(\cdot)$ is at most polynomially growing. Now by H\"{o}lder inequality,
\begin{equation*}
\EE_{(t,x,z)}\int_t^T \mc{P}(Z_s) \left(X_s^\pi\right)^\gamma \ud s \leq \left(\EE_{(t,x,z)}\int_t^T \mc{P}^q(Z_s) \ud s\right)^{1/q}\left(  \EE_{(t,x,z)}\int_t^T\left(X_s^\pi\right)^\gamma p\ud s\right)^{1/p}
\end{equation*}
with $pr = 2$ and $1/p + 1/q = 1$, it reduces to show $X_\cdot^\pi \in L^2([0,T]\times \Omega)$ uniformly in $\delta$. This can be done in a similar manner as done in Section~\ref{sec_wealth}: define $f^\delta(t,x,z;s)$ as the conditional expectation, derive the PDE satisfied for $f^\delta$, making a Anzarts of the form $x^2e^{A^\delta(s-t)z + B^\delta(s-t)}$ and solve a slightly different Riccati equation for $A^\delta(\tau)$ and $B^\delta(\tau)$. Following straightforward but length analysis, it can be verified that $A^\delta$ and $B^\delta$ are uniformly bounded in $\delta$, so does the second moment of $X_\cdot^\pi$. 

(b) $0 < \alpha < 1/4$. In addition to $\vz$ and $\vo$ given above, we solve $\widetilde v^{2\alpha}$ from \eqref{eq_vtat}:
\begin{equation}
\widetilde v^{2\alpha} = -\frac{x^\gamma}{2(1-\gamma)}\mu^2(T-t)ze^{\frac{\mu^2\gamma}{2(1-\gamma)}z(T-t)}.
\end{equation}
Still, every quantity is of the form \eqref{eq_exform} and the argument repeats the previous case.

(c) $\alpha = 1/4$. This is the critical case and similar to case (b). Solving $\vot$ from \eqref{eq_vot}, we obtain $\vot = \vo + \widetilde v^{2\alpha}$. Therefore, all quantities again can be viewed in the form of \eqref{eq_exform}, and the uniform boundedness in $\delta$ follows from H\"{o}lder inequality and finite second moment of $X_\cdot^\pi$.

An example that $\pz \not\equiv \pzt$ could also be validated in a similar manner. For instance, if we choose 
\begin{equation*}
\pzt = c\pz = \frac{\mu cxz}{1-\gamma}, \;c>0 \quad \text{ and } \quad \pot = \pz = \frac{\mu xz}{1-\gamma},
\end{equation*}
The only changes in the previous argument are: in Assumption \ref{assump_piregularity} (iii), we prove that $\widetilde X_s$ possesses full support $\RR^+$ by substituting $\pzt$ into \eqref{eq_Xttilde} and solving the SDE to which the solution is a geometric Brownian Motion; and in Assumption \ref{assump_optimality} (ii), we first solve $\vzt$ from \eqref{eq_vzt}.
}

\section{Conclusion}\label{sec_conclusion}

In this paper, we have considered the portfolio allocation problem in the context of a slowly varying stochastic environment and when the investor tries to maximize her terminal utility in a general class of utility functions. We proved that the zeroth order strategy identified in \cite{FoSiZa:13} is in fact asymptotically optimal up to the first order within a specific class of strategies. We have made precise the assumptions needed in order to rigorously establish this asymptotic optimality. These assumptions are on the coefficients of the model, on the utility function, and on the zeroth order value function, that is the solution to the classical Merton problem with constant coefficients. Finally, we analyzed a fully solvable example in order to demonstrate that our assumptions are reasonable.

In an ongoing work, \cite{FoHu:XX}, we are establishing the same type of results in the case of a fast varying stochastic environment, and, ultimately, in the case of a model with two factors, one slow and one fast. We also plan to analyze the effect of the first order correction in the strategy on the second order correction of the value function. 

Our analysis deals with classical solutions of the partial differential equations involved in the problem, namely, the Merton PDE for the leading order value functions and linear equations with source for the correction terms. A full optimality result would require working with viscosity solutions of the HJB of the full problem, and that is also part of our future research.


\appendix
\section{Proof of Proposition \ref{prop_U}}\label{app_U}
Proof of (i). Without loss of generality, we assume $E = [a,b] \subset [0,1)$. Notice that $a$ can be zero, but $b$ is strictly less than 1. Define $f(x,y) = x^y$, since $f_x^{(7)}(x,y)$ is continuous in $[x_0-\delta, x_0+\delta]\times E$, $\forall x_0 \in (0,\infty)$ and $f_x^{(7)}(x,y)$ is integrable on $E$,
\begin{equation*}
U(x) = \int_E f(x,y) \nud{y}
\end{equation*}
is $C^7(0,\infty)$,  Moreover, we have
\begin{equation}\label{eq_Uprime}
U^{(i)}(x) = \int_E f_x^{(i)}(x,y)\nud{y}, \text{for } i \leq 7.
\end{equation}
The monotonicity and concavity follows by the sign of $U'(x)$ and $U''(x)$ in \eqref{eq_Uprime}. $U(0+) = 0$ follows by Dominated Convergence Theorem (DCT). We have:
\begin{align*}
&\lim_{x\to 0+} U'(x) \geq \lim_{x\to 0+} \int_{a+\delta}^b yx^{y-1} \nud{y} \geq \lim_{x\to 0+}(a+\delta)\left(\frac{1}{x}\right)^{1-b} \nu([a+\delta,b]) = +\infty, \text{for a given } \delta; \\
&\lim_{x\to+\infty} U'(x) = \lim_{x\to+\infty} \int_a^b y\left(\frac{1}{x}\right)^{1-y}\nud{y} \leq \lim_{x\to+\infty} b\left(\frac{1}{x}\right)^{1-b}\nu([a,b]) = 0; \\
& AE[U] = \lim_{x\to+\infty} \frac{x\int_a^byx^{y-1}\nud{y}}{\int_a^b x^y\nud{y}} = \lim_{x\to+\infty} \frac{\int_a^byx^{y}\nud{y}}{\int_a^b x^y\nud{y}} \leq b < 1,
\end{align*}
which shows the Inada and Asymptotic Elasticity conditions \eqref{eq_usualconditions}. To show Assumption~\ref{assump_U}~\eqref{assump_Urisktolerance} is satisfied, we follow Remark \ref{rem_U}, and prove the following: a) $R(0) = 0$, $R(x)$ is strictly increasing on $[0, \infty)$; and b) $\abs{R^j(x)\left(\partial_x^{(j+1)}R(x)\right)} \leq K$, $\forall 0\leq j \leq 4$.
For convenience, we introduce the short-hand notation
\begin{equation*}
\average{f(y)}_x = \int_E f(y)x^y\nud{y},
\end{equation*}
and in the sequel, we shall omit the subscript $x$ when there is no confusion. 

Following from \eqref{eq_Uprime} and using the short-hand notation, $R(x)$ is given by
\begin{equation}\label{eq_R}
R(x) = \frac{\int_E yx^{y-1}\nud{y}}{\int_E y(1-y)x^{y-2}\nud{y}} = x \frac{\average{y}}{\average{y(1-y)}}.
\end{equation}
Since $1-y$ is bounded by $1-b$ and $1-a$, we deduce
\begin{equation}\label{eq_Rbounds}
\frac{x}{1-a} \leq R(x)\leq \frac{x}{1-b},
\end{equation} and obtain $R(0) = 0$ by letting $x\to 0$. Taking derivative in \eqref{eq_R} gives
\begin{equation*}
R'(x) = \frac{\average{y(y+1)}\average{y(1-y)} - \average{y}\average{y^2(1-y)}}{\average{y(1-y)}^2}.
\end{equation*}
The positiveness of $R'(x)$ on $[0,\infty)$ follows by
\begin{align*}
R'(x) &= \frac{\average{y}\average{y^3}+\average{y}^2 - \average{y^2}^2 - \average{y}\average{y^2}}{\average{y(1-y)}^2} \geq \frac{\average{y}^2-\average{y}\average{y^2}}{\average{y(1-y)}^2}\\
&= \frac{\average{y}}{\average{y(1-y)}} \geq \frac{\average{y}}{(1-a)\average{y}} = \frac{1}{1-a},
\end{align*}
where we have used  $\average{y}\average{y^3}\geq \average{y^2}^2$. Thus, $R'(x)$ is bounded below by $\frac{1}{1-a}$ on  $[0,\infty)$, and consequently, $R(x)$ is strictly increasing for $x \geq 0$. To show $R'(x) < K$, we derive the upper bound as follows:
\begin{equation}\label{eq_Rp}
R'(x) \leq \frac{(b+1)(1-a)\average{y}^2}{(1-b)^2\average{y}^2} = \frac{(b+1)(1-a)}{(1-b)^2}.\\
\end{equation}
To show $\abs{R(x)R''(x)} \leq K$, we first compute $R''(x)$:
\begin{align*}
R''(x) &= \frac{1}{x}\left(\frac{\average{y^2(y+1)}}{\average{y(1-y)}} + \frac{2\average{y^2(1-y)}^2\average{y}}{\average{y(1-y)}^3}- \frac{\average{y^2}\average{y^2(1-y)} + \average{y}\average{y^3(1-y)} + \average{y^2(1-y)}\average{y(y+1)}}{\average{y(1-y)}^2}\right).
\end{align*}
Then, the upper bound and lower bound of $R(x)R''(x)$ are computed as follows:
\begin{align*}
&R(x)R''(x) \leq  \frac{\average{y}}{\average{y(1-y)}}\left(\frac{(b+1)\average{y^2}}{(1-b)\average{y}} + \frac{2(1-b)^2\average{y^2}^2\average{y}}{(1-b)^3\average{y}^3}\right) \leq \frac{1}{1-b} \frac{b+3}{1-b}; \\
&R(x)R''(x) \geq - \frac{\average{y}}{\average{y(1-y)}} \frac{(1-a)\average{y^2}^2 + (1-a)\average{y}\average{y^3} + (1-a)(1+b)\average{y^2}\average{y}}{(1-b)^2\average{y}^2} \geq -\frac{(b+3)}{(1-b)^2}.
\end{align*}
Combine the above bounds, we have the desired results $\abs{R(x)R''(x)} \leq K$. Similar arguments works for $j = 2, 3, 4$ by straightforward calculations. 

The last step is to show the growth condition of $I(y) = U'^{(-1)}(y)$. For $y \geq U'(1) = \int_E y\nud{y}$, we have $I(y) \leq 1$. The other case $y = U'(x) \leq U'(1)$, where $x\geq 1$,
\begin{align*}
&U'(x) \leq x^{b-1} \int_E y\nud{y} , \quad \forall x \geq 1, \\
&\Rightarrow \quad U'\left(\left(\frac{t}{\int_E y\nud{y}}\right)^{1/(b-1)}\right) \leq t, \quad \forall t \leq \int_E y\nud{y}\\
& \Rightarrow \quad I(y) \leq \kappa y^{-\gamma}, \quad \forall y \leq \int_E y\nud{y},
\end{align*}
where $\kappa= \left(\frac{1}{\int_E y\nud{y}}\right)^{1/(b-1)}$ is a constant depending solely on $\nud{y}$, and $\alpha = \frac{1}{1-b} > 1$. Combining the two cases, we have $I(y) \leq \alpha + \kappa y^{-\alpha}$.

Proof of (ii). This class of utility functions is defined via the inverse of marginal utility $I(y)$, where $U(x)$ can be recovered by:
\begin{equation}
U(x) = \int_0^x U'(t) \ud t = \int_0^x I^{(-1)}(t)\ud t.
\end{equation}
Then $U(0+)=0$ is automatically satisfied. By definition of $I(y)$, $I(y) \in C^\infty(0,\infty)$, so does $U(x)$. 
The strictly monotonicity and strictly concavity are given by:
\begin{align*}
&U'(x) = I^{(-1)}(x)>0,
&U''(x) = \frac{1}{I'(I^{(-1)}(x))} = \left(\int_0^N -s\left(I^{(-1)}(x)\right)^{-s-1} \nud{s} \right)^{-1}<0.
\end{align*}
By DCT, one has
\begin{align*}
&I(+\infty) = \lim_{y\to+\infty} \int_0^N y^{-s} \nud{s} = 0, \\
&I(0) = \lim_{y\to0}\int_0^N y^{-s}\nud{s} \geq \lim_{y\to0}\int_\delta^N y^{-s}\nud{s} \geq \lim_{y\to0} y^{-\delta}\nu[\delta,N] = +\infty,\\
& AE[U] = \lim_{x\to+\infty} x\frac{U'(x)}{U(x)} = \lim_{x\to+\infty}\frac{U'(x) + xU''(x)}{U'(x)} = \lim_{x\to+\infty} 1 - \frac{x}{R(x)} = 1 - \lim_{x\to+\infty}\frac{1}{R'(x)} \leq 1-\frac{1}{N},
\end{align*}
where we have used the fact that $R'(x) \leq N$ derived below.

\done{To show that condition \eqref{assump_Uiii} is satisfied, we follow from the idea in \cite[Example 18]{KaZa:14}, which gives the proof for $j=1,2$ under constant $\lambda$.}
From Proposotion \ref{prop_H}, $H(x,T,\lambda(z)) = I(e^{-x}) = \int_0^N e^{xs}\nud{s}$, and the risk tolerance $R(x)$ is given by:
\begin{align*}
R(x) = -\frac{U'(x)}{U''(x)} = H_x(H^{(-1)}(x,T,\lambda(z)),T,\lambda(z)) = \int_0^N se^{H^{(-1)}(x,T,\lambda(z))s}\nud{s}.
\end{align*}
The fact that $R(0) = 0$ follows by DCT and $H^{(-1)}(0+,T,\lambda(z)) = -\infty$. $R(x)$ has bounded derivative, since:
\begin{align*}
R'(x) = \frac{H_{xx}(H^{(-1)}(x,T,\lambda(z)),T,\lambda(z))}{H_x(H^{(-1)}(x,T,\lambda(z)),T,\lambda(z))} = \frac{\int_0^N e^{H^{(-1)}(x,T,\lambda(z))s}s^2 \nud{s}}{\int_0^N e^{H^{(-1)}(x,T,\lambda(z))s}s\nud{s}} \leq N.
\end{align*}
$R(x)$ is strictly increasing, since the numerator stay positive for $x>0$, i.e. $0 < R'(x) \leq N$. To show $\abs{R(x)R''(x)} \leq K$, one needs
\begin{align*}
R(x)R''(x) + \left(R'(x)\right)^2 = \frac{1}{2}(R^2(x))'' = \frac{H_{xxx}(H^{(-1)}(x,T,\lambda(z)),T,\lambda(z))}{H_x(H^{(-1)}(x,T,\lambda(z)),T,\lambda(z))} = \frac{\int_0^N e^{H^{(-1)}(x,T,\lambda(z))s}s^3 \nud{s}}{\int_0^N e^{H^{(-1)}(x,T,\lambda(z))s}s\nud{s}} \leq N^2.
\end{align*}
Since $(R^2(x))''$ and $R'(x)$ are bounded, so does $R(x)R''(x)$. Similar arguments works for $R^j(x)\left(\partial_x^{(j+1)}R(x)\right)$ with $j = 2, 3, 4$ by using the following identities:
\begin{align*}
&R^2R'''+ R'^3 + 4RR'R''= \frac{\partial_x^4H(H^{(-1)}(x,T,\lambda(z)),T,\lambda(z))}{H_x(H^{(-1)}(x,T,\lambda(z)),T,\lambda(z))} = \frac{\int_0^N e^{H^{(-1)}(x,T,\lambda(z))s}s^4 \nud{s}}{\int_0^N e^{H^{(-1)}(x,T,\lambda(z))s}s\nud{s}} \leq N^3
\\
&R^3R^{(4)} + 7R^2R'R''' + R'^4 + 11RR'^2R'' + 4R^2R''^2 = \frac{\partial_x^5H(H^{(-1)}(x,T,\lambda(z)),T,\lambda(z))}{H_x(H^{(-1)}(x,T,\lambda(z)),T,\lambda(z))} \leq N^4 , \\
&12R'^5 + 32R^2R'^2R''' + 57RR'^3R'' + R^4R^{(5)} + 15R^3R''R''' + 11R^3R'R^{(4)} + 54R^2R'R''^2 \\
& \qquad = \frac{\partial_x^6H(H^{(-1)}(x,T,\lambda(z)),T,\lambda(z))}{H_x(H^{(-1)}(x,T,\lambda(z)),T,\lambda(z))} \leq N^5 .
\end{align*}

Notice that $I(y)$ satisfies the polynomial growth condition due to the following: denote by $\kappa = \nu([0,N])$, if $y \geq 1$, $I(y) \leq \kappa$, otherwise when $y < 1$
\begin{align*}
I(y) \leq \kappa y^{-N}.
\end{align*}
Therefore, combining the two cases and defining $\alpha = \max\{N,\kappa\}$ yields the inequality \eqref{cond_I}.

\section{Proof of Proposition \ref{prop_ssh}}\label{app_ssh}
\done{For the case $j=0,1$, results under constant $\lambda$ are presented in \cite[Proposotion 14]{KaZa:14}. We first generalize their results to $\lambda(z)$ for $j=0, 1$, and then give the proof for $j = 2, 3, 4$.}
To show \eqref{eq_prop_Rbounds}, we use the relation \eqref{eq_rzH} between the risk tolerance function $R(t,x;\lambda(z))$ and the function $H(x,t,\lambda(z))$ (which is defined in Proposition \ref{prop_H}), namely 
\begin{equation}\label{eq_prop_rzH}
R(t,x;\lambda(z)) = H_x(H^{(-1)}(x,t,\lambda(z)),t,\lambda(z)).
\end{equation}
Differentiating \eqref{eq_prop_rzH} with respect to $x$, and letting $t = T$ produces:
\begin{equation*}
R'(x) = \left.\frac{H_{xx}(y,T,\lambda(z))}{H_x(y,T,\lambda(z))}\right|_{y = H^{(-1)}(x,T,\lambda(z))}.
\end{equation*}
Using $R'(x) \leq C = \sqrt{K/2}$ (proved in Proposition \ref{lem_U}) gives, for all $x,z\in \RR$,
\begin{equation*}
\abs{H_{xx}(x,T,\lambda(z))} \leq \sqrt{K/2} H_x(x,T,\lambda(z)).
\end{equation*}
Notice that for fixed $\lambda(z)$,  both $H_{xx}(x,t,\lambda(z))$ and $H_x(x,t,\lambda(z))$ satisfy the heat equation \eqref{eq_H}. Comparison Principle ensures that the inequality is preserved for $t < T$, i.e. for $(x,t,z)\in \RR \times[0,T]\times\RR$,
\begin{equation*}
\abs{H_{xx}(x,t,\lambda(z))} \leq \sqrt{K/2} H_x(x,t,\lambda(z)).
\end{equation*}
Using \eqref{eq_prop_rzH} again, we obtain:
\begin{equation}\label{eq_prop_case0}
\partial_xR(t,x;\lambda(z)) = \left.\frac{H_{xx}(y,t,\lambda(z))}{H_x(y,t,\lambda(z))}\right|_{y = H^{(-1)}(x,t,\lambda(z))} \leq \sqrt{K/2} := K_0.
\end{equation}
Thus, we have shown \eqref{eq_prop_Rbounds} with $j=0$. 

To complete the proof in the case $j = 1$, we first obtain a relation between derivatives of $R(t,x;\lambda(z))$ and derivatives of $H(x,t,\lambda(z))$:
\begin{align}\label{eq_prop_case1}
R_x^2(t,x;\lambda(z)) + RR_{xx}(t,x;\lambda(z)) =  \frac{1}{2}\partial_x^{(2)}R^2(t,x;\lambda(z))= \left.\frac{H_{xxx}(y,t,\lambda(z))}{H_x(y,t,\lambda(z))}\right|_{y = H^{(-1)}(x,t,\lambda(z))}.
\end{align}
Let $t = T$ in the above identity, then, the middle quantity is reduced to $\frac{1}{2}\partial_x^{(2)} R^2(x)$ and is bounded by $K/2$ as assumed in \eqref{assump_Uiii}, so does the ratio of $H_{xxx}(x,T,\lambda(z))$ over $H_x(x,T,\lambda(z))$ for all $x \in \RR$. Standard Comparison Principle applies and the ratio remains bounded for $t < T$. This results in the boundedness of $\abs{R_x^2(t,x;\lambda(z)) + RR_{xx}(t,x;\lambda(z))}$. Combining with \eqref{eq_prop_case0}, we achieve:
\begin{equation*}
\abs{RR_{xx}(t,x;\lambda(z))} \leq K/2 + K_0^2 := K_1.
\end{equation*}

To deal with $j=2$, we first obtain the identity:
\begin{equation}\label{eq_prop_case2}
R^2R_{xxx}(t,x;\lambda(z))+ R_x^3(t,x;\lambda(z)) + 4RR_xR_{xx}(t,x;\lambda(z))= \left.\frac{\partial_x^{(4)}H(y,t,\lambda(z))}{H_x(y,t,\lambda(z))}\right|_{y = H^{(-1)}(x,t,\lambda(z))}.
\end{equation}
At terminal time T, each term on the left-hand side is bounded (cf. Remark \ref{rem_U}), therefore, the right-hand side is bounded. Then a similar argument based on Comparison Principle gives the following estimate:
\begin{equation*}
\abs{R^2R_{xxx}(t,x;\lambda(z))} \leq K_2.
\end{equation*}

The remaining cases are completed by  replacing \eqref{eq_prop_case2} by the following
\begin{align*}
&R^3R^{(4)}(t,x;\lambda(z)) + 7R^2R_xR_{xxx}(t,x;\lambda(z)) + R_x^4(t,x;\lambda(z)) + 11RR_x^2R_{xx}(t,x;\lambda(z)) + 4R^2R_{xx}^2(t,x;\lambda(z)) \\
& \qquad = \left.\frac{\partial_x^{(5)}H(y,t,\lambda(z))}{H_x(y,t,\lambda(z))}\right|_{y = H^{(-1)}(x,t,\lambda(z))}, \\
&12R_x^5(t,x;\lambda(z)) + 32R^2R_x^2R_{xxx}(t,x;\lambda(z)) + 57RR_x^3R_{xx}(t,x;\lambda(z)) + R^4R^{(5)}(t,x;\lambda(z)) + 15R^3R_{xx}R_{xxx}(t,x;\lambda(z)) \\
& \qquad + 11R^3R_xR^{(4)}(t,x;\lambda(z)) + 54R^2R_xR_{xx}^2(t,x;\lambda(z)) =\left.\frac{\partial_x^{(6)}H(y,t,\lambda(z))}{H_x(y,t,\lambda(z))}\right|_{y = H^{(-1)}(x,t,\lambda(z))},
\end{align*}
and repeating the same argument.

The bounds $\widetilde K_j$ are easily obtained by expanding $\partial_x^{(j)} R^j(t,x;\lambda(z))$ and using \eqref{eq_prop_Rbounds}. To show \eqref{eq_prop_Rbound}, we notice that $R(t,x;\lambda(z))$ is shown to be strictly increasing in Proposition \ref{prop_mono}, and $R'(t,x;\lambda(z)) \leq K_0$, integrating on both sides with respect to $x$ yields the desired result.
\section{Assumptions in Section \ref{sec_accuracy2}}\label{appendix_addasump}

This set of assumptions is used in deriving Proposition \ref{prop_piaccuracy} where we establish the accuracy of approximation of $\Vzt$ summarized in Table \ref{tab_accuracy}. 
\begin{assump}\label{assump_optimality}
	Let $\MCA_0(t,x,z)\left[\pzt,\pot,\alpha\right]$ be the family of trading strategies defined in \eqref{def_A0}. Recall that $X^\pi$ is the wealth generated by the strategy $\pi=\pzt+\delta^\alpha\pot$ as defined in \eqref{def_Xtopt}. In order to condense the notation,  we systematically omit  the argument $(s,X_s^\pi,Z_s)$ in what follows. According to the different cases, we further require:
	\begin{enumerate}[(i)]
		\item\label{assump_optimality_eq} If $\pzt \equiv \pz$,
		\begin{enumerate}
			\item\label{assump_optimality_eqmore} If $\alpha > 1/4$, the following quantities are uniformly bounded in $\delta$:
			
			$\quad\EE_{(t,x,z)}\int_t^T \MCM\vz\ud s $,
			$\EE_{(t,x,z)}\int_t^T \MCM\vo\ud s$,
			$\EE_{(t,x,z)}\int_t^T\sigma^2(Z_s)\left(\pot\right)^2\vz_{xx} \ud s $,
			
			$\quad\EE_{(t,x,z)}\int_t^T\sigma^2(Z_s)\left(\pot\right)^2\vo_{xx} \ud s$,
			$\EE_{(t,x,z)}\int_t^T\sigma^2(Z_s)\pz\pot\vo_{xx} \ud s $,
			$\EE_{(t,x,z)}\int_t^T\mu(Z_s)\pot\vo_x \ud s $,
			
			$\quad\EE_{(t,x,z)}\int_t^T g(Z_s)\sigma(Z_s)\pz\vo_{xz} \ud s$,
			$\EE_{(t,x,z)}\int_t^Tg(Z_s)\sigma(Z_s)\pot\vz_{xz} \ud s$,
			$\EE_{(t,x,z)}\int_t^T g(Z_s)\sigma(Z_s)\pot\vo_{xz} \ud s$.

			Here, recall that $\vz$ and $\vo$ are the leading order term and first order correction of $\Vz$ as well as of $\Vzl$, and they satisfy \eqref{eq_vz} and \eqref{eq_vo} respectively.  
			\item\label{assump_optimality_eqless} In the case $0 < \alpha < 1/4$, if $(t,x,z) \in \MCK_1$, we need 
			
			$\quad\EE_{(t,x,z)}\int_t^T \MCM\vz\ud s $,
			$\EE_{(t,x,z)}\int_t^T \MCM\vtat\ud s $,
			$\EE_{(t,x,z)}\int_t^T \sigma^2(Z_s)\left(\pot\right)^2\vtat_{xx} \ud s $,
			
			$\quad\EE_{(t,x,z)}\int_t^T \sigma^2(Z_s)\pz\pot\vtat_{xx} \ud s $,
			$\EE_{(t,x,z)}\int_t^T \mu(Z_s)\pot\vtat_x\ud s $,
			$\EE_{(t,x,z)}\int_t^T g(Z_s)\sigma(Z_s)\pz\vz_{xz}\ud s $,
			
			$\quad\EE_{(t,x,z)}\int_t^T g(Z_s)\sigma(Z_s)\pz\vtat_{xz}\ud s $,
			$\EE_{(t,x,z)}\int_t^T g(Z_s)\sigma(Z_s)\pot\vz_{xz}\ud s $,
			$\EE_{(t,x,z)}\int_t^T g(Z_s)\sigma(Z_s)\pot\vtat_{xz}\ud s $,
			to be uniformly bounded in $\delta$, where $\vtat$ is the coefficient of $\delta^{2\alpha}$ in the expansion of $\Vzt$ in the case $ 0 < \alpha < 1/4$, and satisfies the linear PDE \eqref{eq_vtat}.
			
			Otherwise if $(t,x,z) \in \MCC_1$, we only need
			
			$\quad \EE_{(t,x,z)}\int_t^T \MCM\vz\ud s$,
			$\EE_{(t,x,z)}\int_t^T \MCM\vo\ud s $,
			$\EE_{(t,x,z)}\int_t^T g(Z_s)\sigma(Z_s)\pz\vo_{xz}\ud s $,
			to be uniformly bounded.
			\item\label{assump_optimality_eqequal}For the critical case $\alpha = 1/4$, we need all the above assumptions in part \eqref{assump_optimality_eqless} with
			$\alpha$ replaced by $1/4$, except the sixth one which is not needed.
			Then, the term $\vtat$ becomes $\vot$, which is the coefficient of $\sqrt\delta$ in the expansion and satisfies \eqref{eq_vot}.
		\end{enumerate}
		\item\label{assump_optimality_neq} If $\pzt \not\equiv \pz$,
		\begin{enumerate}[(a)]
			\item For $(t,x,z)$ in the region $\MCK$, the following quantities need to be uniformly bounded in $\delta$:
			
			$\quad\EE_{(t,x,z)}\int_t^T \MCM \vzt \ud s $,
			$\EE_{(t,x,z)}\int_t^T g(Z_s)\sigma(Z_s)\left(\pzt+ \delta^\alpha\pot\right)^2\vzt_{xz} \ud s$,
			$\EE_{(t,x,z)}\int_t^T \sigma^2(Z_s)\left(\pot\right)^2\vzt_{xx} \ud s$,
			
			$\quad\EE_{(t,x,z)}\int_t^T \sigma^2(Z_s)\pzt\pot\vzt_{xx} \ud s $,
			$\EE_{(t,x,z)}\int_t^T \mu(Z_s)\pot\vzt_x\ud s$,
						where $\vzt$ is the leading order solution that satisfies \eqref{eq_vzt}.
			\item For $(t,x,z)$ in the region $\MCC$, where $\pzt$ and $\pz$ are identical, requirements are the same as in part \eqref{assump_optimality_eq} with $(t,x,z)$  restricted to be in $\MCC$.
		\end{enumerate}
	\end{enumerate}
\end{assump}

\bibliographystyle{plainnat}
\bibliography{Reference}

\end{document}